\documentclass[aps,pra,twocolumn,floatfix,nofootinbib,superscriptaddress,graphics,longbibliography]{revtex4-2}
\usepackage{graphicx,mathpazo,ragged2e,dcolumn,xcolor,dsfont,physics,bm,braket,lipsum,tikz,amsmath, amsthm, amsfonts, amssymb, bbm,enumitem,mathtools,pifont,lmodern}
\usepackage[most]{tcolorbox}
\newlist{todolist}{itemize}{2}
\setlist[todolist]{label=$\square$}

\definecolor{alered}{RGB}{222,94,100}
\definecolor{alecolor}{RGB}{198,113,190}
\definecolor{equationcolor}{RGB}{222,94,100}

\newtcolorbox[auto counter]{mybox}[2][]{%
    breakable,
    enhanced,
    sharp corners,
    colback=violet!3!white,
    colframe=violet!40!white,
    fonttitle=\bfseries,
    title={\centering \strut #2}, 
    enlarge bottom at break by=5mm,
    enlarge top at break by=5mm,
    overlay first={%
        \draw[black, line width=0.5mm](frame.south west)--(frame.south east);
        \node[anchor=north east] at (frame.south east) {continued on next page};
    },
    overlay middle={%
        \draw[black, line width=0.5mm](frame.south west)--(frame.south east);
        \draw[black, line width=0.5mm](frame.north west)--(frame.north east);
        \node[anchor=north east] at (frame.south east) {continued on next page};
        \node[anchor=south west] at (frame.north west) {continued from next page};
    }, overlay last={%
        \draw[black, line width=0.5mm](frame.north west)--(frame.north east);
        \node[anchor=south west] at (frame.north west) {continued from next page};},
    #1
}

\usepackage[colorlinks=true,linkcolor=alecolor,citecolor=alecolor,urlcolor=alecolor]{hyperref}

\renewcommand{\v}[1]{\ensuremath{\boldsymbol #1}}
\newcommand{\ms}[1]{\textsf{#1}}
\newcommand{\1}{\mathbbm{1}}

\newcommand{\iden}{\mathbbm{1}}
\DeclareMathOperator{\diag}{diag}

\newtheorem{theorem}{Theorem}[section]
\newtheorem{corollary}{Corollary}[section]
\newtheorem{proposition}{Proposition}[section]
\newtheorem{lemma}[theorem]{Lemma}
\theoremstyle{definition}

\let\oldaddcontentsline\addcontentsline
\renewcommand{\addcontentsline}[3]{}

\begin{document}

\title{Reconstructing thermal states using dimensionally limited probes : \\
\texttt{\small A Model for Limited Control \& Memory in Quantum Thermodynamics}}

\author{Jake Xuereb}
\affiliation{Vienna Center for Quantum Science and Technology, Atominstitut, TU Wien, 1020 Vienna, Austria}
\email{jake.xuereb@tuwien.ac.at}

\author{A. de Oliveira Junior}
\affiliation{Center for Macroscopic Quantum States (bigQ), Department of Physics, Technical University of Denmark, 2800 Kongens Lyngby, Denmark}

\author{Fabien Clivaz}
\affiliation{Vienna Center for Quantum Science and Technology, Atominstitut, TU Wien, 1020 Vienna, Austria}
\affiliation{Fachhochschule Technikum Wien, Höchstädtplatz 6, 1200 Wien}

\author{Pharnam Bakhshinezhad}
\affiliation{Vienna Center for Quantum Science and Technology, Atominstitut, TU Wien, 1020 Vienna, Austria}

\author{Marcus Huber}
\affiliation{Vienna Center for Quantum Science and Technology, Atominstitut, TU Wien, 1020 Vienna, Austria}
\affiliation{Institute for Quantum Optics and Quantum Information - IQOQI Vienna, Austrian Academy of Sciences, Boltzmanngasse 3, 1090 Vienna, Austria}
\date{November 6\textbf{}, 2025}

\begin{abstract}
Whilst the complexity of acquiring knowledge of a quantum state has been extensively studied in the fields of quantum tomography and quantum learning, a physical understanding of its operational role and cost in quantum thermodynamics is lacking. Knowledge is central to thermodynamics, as exemplified by Maxwell's demon thought experiment, where a demonic agent is able to extract paradoxical amounts of work -- reconciled by the thermodynamic costs of acquiring this knowledge. In this work, we address this gap by extending unitary models of measurement to incorporate the resources available to an agent. We view an agent's knowledge of a quantum state as their ability to reconstruct it unitarily given access to states with partial knowledge of the \textit{true} state. In our model, an agent correlates an unknown $d$-dimensional system, with copies of a $k$-dimensional probe ($k\le d$), which are then used to unitarily reconstruct an \textit{estimate} state in $d$-dimensional memories. We find that this framework is a unitary representation of coarse-grained POVMs. As an application, we investigate the role of knowledge in an extended Szilard Engine scenario.
\end{abstract}

\maketitle

\section{Introduction \label{Sec:introduction}}

In classical thermodynamics~\cite{fermi, callen}, knowledge is at the heart of an agent's ability to manipulate a system. An agent controlling an ideal gas within a piston, while having access only to coarse-grained quantities such as pressure and volume, faces limitations in their capacity to transform the system. For instance, extracting less work compared to Maxwell's demon~\cite{maxwell1867, vedral_review,junior2025friendlyguideexorcisingmaxwells}. Indeed, if the agent possesses complete knowledge of both the position and velocity of each particle in the gas, they can extract work seemingly for free. However, when the costs associated with acquiring, storing, and erasing the knowledge of each particle are considered, this ``\textit{seemingly free work}" is revealed to come at a price~\cite{landauer1961, bennett1982}. This classical lesson, that knowledge has operational and thermodynamic costs, motivates our exploration of a quantum treatment.

In quantum thermodynamics~\cite{Goold_2016,binder2019thermodynamics}, the operational role of knowledge is yet to be clarified. Firstly, from an operational point of view, it is not clear how the knowledge an agent possesses limits their ability to achieve tasks such as extracting work or erasing information. Secondly, from a foundational perspective, it is not clear what the cost of being able to reconstruct a thermal state is with access to some quantum resources such as correlating unitaries and a probe system of a given dimension. This is of particular importance when considering \textit{agentic} questions in thermodynamics as we must take into account the cost for an agent acquiring knowledge of a state before being able to act on it e.g. extracting work from this state.

In this work, we attempt to address these foundational and operational aspects of knowledge in quantum thermodynamics. More precisely, we:
\begin{itemize}
    \item[Sec III] provide a model which quantifies the resources required (probe dimension, correlating unitaries) to acquire knowledge by recasting state estimation as a fully unitary, thermodynamically consistent protocol.  
    \item[Sec IV] Explore how with access to multiple rounds of interaction an agent can acquire multiple states which approximate the state of the system and improve the quality of their knowledge by acting on these estimate states.
    \item[Sec  V] Investigate how an agent’s knowledge (as constrained by probe dimension and allowed unitaries) limits their ability to perform thermodynamic tasks. As a first example we consider an extended Szilard engine model where the agent's knowledge of particle's position limits their ability to extract work.  
\end{itemize}   

The cost of acquiring (classical) knowledge of a quantum state has been extensively studied in the emerging field of quantum state learning~\cite{Anshu2024}, as well as in the contexts of quantum tomography~\cite{keyl_2001,harrow_17,odonnell_tomography} and state estimation~\cite{popescu_massar_95}. In this work, we approach this question from an altogether new operational perspective by formulating state estimation in a unitary and thermodynamically consistent setting. While we focus on reconstructing (diagonal) thermal states, the methods we develop are applicable to reconstructing the diagonal of any quantum state in a fixed basis. Due to the no-cloning theorem, when working on single copy input states, one cannot expect to reconstruct more than the diagonal distribution in a specific basis in single-shot. Quantum tomography~\cite{d2003quantum} and quantum learning~\cite{Anshu2024} provide a sample-complexity theoretic understanding of how hard it is to acquire a classical representation of a quantum state. However, they do not address the resources an agent must make use of to make use of this classical representation to reconstruct the quantum state on a quantum memory. In this work we posit that an agent’s knowledge of a quantum state reflects their ability to unitarily reconstruct the system’s diagonal (in a desired basis) in a memory system they have access to, given specified resources such as probe dimension and allowed unitaries.

\begin{figure*}[t]
     \centering
     \includegraphics[width = \linewidth]{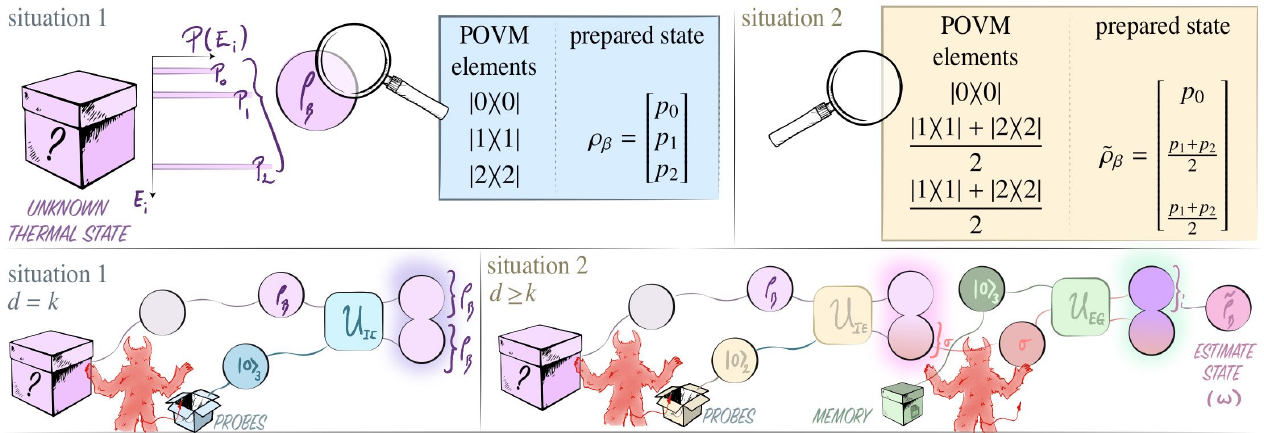}     
     \caption{\emph{Summary}. The main message of this work is the reformulation of a scenario where an agent with access to only limited POVMs attempts to estimate the state of an unknown system, as a unitary protocol allowing us to connect to thermodynamics. Using lower-dimensional probes, an agent creates correlations between their probe and the system to extract information, which is then used to reconstruct the unknown state in a quantum memory. This creates a contrast between two exemplary situations in state reconstruction. Situation~\hyperref[B:motivating-example]{1}, the agent has access to full POVM elements and is able to perfectly reconstruct a thermal state. Situation~\hyperref[B:motivating-example]{2}, the agent has access to coarse-grained POVMs, allowing them to only partially distinguish between certain energy levels of the system, resulting in an approximate reconstruction of the state. The focus of this work is to study how limited resources such as those in Situation~\hyperref[B:motivating-example]{2} impede an agent's ability to reconstruct the state of an unknown system.}
     \label{fig:summary}
 \end{figure*}

 Formally, we approach this general problem of acquiring knowledge of a quantum state by considering a simple setup. A $d$-dimensional quantum system is prepared in an unknown thermal state. The agent, with access to a $k \leq d$ dimensional quantum system, referred to as the probe system, unitarily couples the probe with the unknown state, creating correlations. Using the information captured in the probe, the agent unitarily reconstructs an estimated state of the system in a $d$-dimensional quantum memory. The agent is then allowed to perform unitaries on the obtained estimates to improve their ability to faithfully represent the state.  Fascinatingly, we observe that this dimensionally-constrained unitary protocol directly corresponds to \textit{measure-and-prepare} channels involving coarse-grained POVMs. Thereby building a bridge from thermodynamics to state estimation as exemplified in Fig.~\ref{fig:summary}. 


\section{Setting the scene \label{S:setting-the-scene}}
We consider a $d$-dimensional system $\ms{S}$, described by a Hamiltonian $H_{\ms{S}}$, prepared in an unknown thermal state $\rho_\beta = e^{-\beta H_{\ms{S}}}/\tr(e^{-\beta H_{\ms{S}}})$ at some, likewise unknown, inverse temperature $\beta = 1/k_B T$, where $k_B$ is the Boltzmann constant. Since this state, and subsequent ones, will be diagonal in the energy eigenbasis, we will adopt a shorthand notation. Specifically, when a density matrix is diagonal in a particular basis (e.g., the energy eigenbasis from which we are extracting information), we will represent the matrix by its diagonal elements as a vector. That is, we write $\rho_{\beta}:=(p_0, \ldots, p_{d-1})$, where the $p_i$ are the diagonal elements of the density matrix. For both the probe and the memory, we represent its state together with its dimension by $\ket{a}_x$ with the index $x$ indicating the dimension.

Knowledge of a state is reflected in an agent's ability to reconstruct it. A canonical approach to this setting involves defining the so-called \textit{measure-and-prepare} (MP) channel $\mathcal{MP}(\cdot)$~\cite{watrous2018theory}, which allows one to frame the task of state estimation as measuring $n$ copies of an unknown state and then preparing other states that serve as \emph{estimates} of this unknown state. Formally, these channels are defined by a set of POVM elements $M_i$ and quantum states $\omega_i$, such that the channel acts on an input state $\rho$ as follows $ \mathcal{MP}(\rho) = \sum_i \omega_i \tr\{M_i \rho\}$. MP channels can be made more physical by relating them to the von Neumann measurement scheme~\cite{peres1995quantum} via a unitary. For an MP channel $\mathcal{MP}(\rho)$, we construct a correlating global unitary that, depending on the state of $\rho$, rotates an ancillary state $\nu_A$ to $\omega_i$. Mathematically, we consider the dilation:
 $\mathcal{MP}(\rho) = \tr_{\ms S}\left\{\sum_i M_i \otimes U_i (\rho \otimes \nu_A) M_i \otimes U^\dagger_i\right\}$. 

Let us discuss two examples to illustrate the estimation of thermal states through a correlating unitary process~(see Fig.~\ref{fig:summary} for a schematic representation of the protocol).

\begin{mybox}{A. Motivating example: full control\label{B:motivating-example}}
Suppose that an agent wants to estimate the state of a qutrit, which is prepared in an (unknown) thermal state $\rho_\beta= (p_0, p_1, p_2)$. If they can perfectly distinguish all three levels with access to POVM elements $M_i = \{\ketbra{0}{0},\ketbra{1}{1},\ketbra{2}{2}\}$, then they are able to estimate the state perfectly with access to a channel which produces energy eigenket $\ket{i}$ whenever $\rho_\beta$ is found in $\ket{i}$.
 \end{mybox}

When the agent's ability to distinguish energy levels is restricted a different scenario occurs.
 
 \begin{mybox}{B. Motivating example: coarse control\label{B:motivating-example-2}}
Suppose now that the agent cannot distinguish the second and third energy levels and has access to coarse grained POVM elements $M'_i = \{\ketbra{0}{0},\frac{\ketbra{1}{1} + \ketbra{2}{2}}{2},\frac{\ketbra{1}{1} + \ketbra{2}{2}}{2}\}$. Differently from the motivating example ~\hyperref[B:motivating-example]{A}, they would correctly prepare $\ket{0}$, but would mix up when to correctly prepare $\ket{1}$ and $\ket{2}$. Consequently, this leads to the estimate state $\widetilde{\rho}_\beta= (p_0, \frac{p_1 + p_2}{2}, \frac{p_1 + p_2}{2})$.
 \end{mybox}

The situation described in the above example physically corresponds to cloning a quantum state in a single basis, which is possible when given a state of equal dimension and equal or greater purity~\cite{scarani2005quantum}. In the present case, a thermal state is diagonal in the energy eigenbasis, so it can be cloned with access to an eigenstate $\ket{q}$ of the Hamiltonian. For example, one can construct a unitary $V$ that, conditioned on the populations of $\rho_\beta$, rotates $\ket{q}$ into different energy eigenstates, leaving $\rho_\beta$ invariant and producing a copy given by $\tr_{\ms S}\{V(\rho_\beta \otimes \ketbra{q}{q})V^\dagger\} = \rho_\beta$, where $V$ is given by
\begin{align}
V =\sum^{d-1}_{j = 0, \, i\neq j} \ketbra{i}{i} \otimes S_{q,i},
\end{align}
with $S_{q,i} ~:=~\sum_{j \neq \{q,i\}} \ketbra{i}{q} + \ketbra{q}{i} + \ketbra{j}{j}$ representing a permutation or a two-level SWAP gate. Conditioned on the unknown system being in state $\ket{i}$, the two-level SWAP gate rotates $\ket{q}$ to $\ket{i}$, effectively cloning the state to the second system. This is possible because the agent has access to a $d$-level measurement probe, which has a dimension equal to that of the system being investigated. This situation is characterized by an agent having \textit{full control}. Or analogously in the state estimation picture, they have access to POVM elements that fully distinguish between each energy level. What if the agent only had access to coarse grained POVM elements? 

As we will show, situations with coarse-distinguishability, such as the second example illustrated in Box~\hyperref[B:motivating-example-2]{B} can be reformulated as a unitary protocol in which an agent attempts to extract information from the system by building correlations with a limited probe system of lower dimension than the system. In Sec.~\ref{S:information-extraction-estimate-generation}, we describe this protocol as a two-step unitary process involving a lower-dimensional probe and an equal-dimensional memory system. In the first step, termed \emph{information extraction} (IE), the low-dimensional probe is correlated with the system of interest using controlled-SWAP unitaries, such that the probe populations are partial sums of the system populations. In the example involving a three-level system, this is accomplished by the following unitary operation:
\begin{equation}
    U_{\text{IE}} =\ketbra{0}{0} \otimes \ketbra{0}{0} \!+\! (\ketbra{1}{1} \!+\! \ketbra{2}{2}) \otimes \sigma_x + \mathbb{1}_{\text{Rest}},
\end{equation}
where in this case $\mathbb{1}_\text{Rest} = \ketbra{01}{01}$ and in general acts as an identity on the untouched subspace. This leads to the following marginal state for the probe
\begin{equation}\label{Eq:example-state}
    \sigma:=\tr_{\ms S}\{U_{\text{IE}}(\rho_{\beta} \otimes \ketbra{0}{0}_2)U_{\text{IE}}^\dagger \},
\end{equation}
where $\sigma = (p_0, p_1 + p_2)$. 

Next, in the \emph{estimate generation} (EG) step, the information captured by the probe is equally distributed following a least assumptions principle in a $d$-dimensional memory system. This is performed by two-level SWAP and quantum Fourier transform operations which ensure that the information is spread correctly and that no quantum coherences accidentally end up in the memory system. 

For instance, consider a probe prepared in the state $\sigma$, given by Eq.~\eqref{Eq:example-state}, coupled to a three-dimensional memory initially in the state $\ket{0}_3$, such that the composite state of the system is $\sigma \otimes \ketbra{0}{0}_3$. By applying a SWAP in the subspace ${\ket{10}, \ket{01}}$, the population of the joint system is mapped to $(p_0, p_1 + p_2, 0, 0, 0, 0)$. Next, applying a Hadamard operation in the subspace spanned by $\ket{01}$ and $\ket{12}$ transforms the state to $p_0 \ketbra{00}{00} + \frac{p_1 + p_2}{2} (\ket{01} + \ket{12})(\bra{01} + \bra{12})$. Finally, tracing out the probe results in the \emph{estimate state} $\omega_0 = \left(p_0, \frac{p_1 + p_2}{2}, \frac{p_1 + p_2}{2}\right)$.
    
Observe that in this protocol, the agent has taken the information obtained from a restricted coupling with their low-dimensional probe and spread it in an agnostic way (relative to their restricted coupling) across a memory system, creating an estimate state. Clearly, different couplings with the lower-dimensional probes generate different estimates. For example, the agent has access to $\binom{3}{2,1} = 3$ possible estimates where $\binom{n}{k_1,k_2} = \frac{n!}{k_1!k_2!}$ is a multinomial coefficient, being them $\omega_1 = \left(\frac{p_0 + p_2}{2}, p_1, \frac{p_0+p_2}{2}\right)$ or $\omega_2 = \left(\frac{p_0+p_1}{2}, \frac{p_0+p_1}{2}, p_2\right)$.

This raises several questions that we address in this work:
\begin{enumerate}
    \item \textit{How does the probe dimension impede reconstruction?} 
    \item \textit{Can globally manipulating multiple estimates allow us to obtain a better reconstruction of this state?}
    \item \textit{How does this partial knowledge constrain an agent's ability to carry out tasks such as work extraction or state transformation?}
\end{enumerate}
In the coming sections, we properly define these questions and shed light on their answers.

\section{Information Extraction \& Estimate Generation \label{S:information-extraction-estimate-generation}}

In the previous section, we discussed the essence of the coarse-grained scenario. We will now formalise this idea by presenting in detail the information extraction and estimate generation steps while introducing the notation.

\subsection{Information Extraction}
If the agent has access to $k < d$ dimensional measurement probes they are only able to couple to the $d$ energy levels by grouping them into $k$ partitions. As a result, they are unable to distinguish between energy levels within a given partition, which impairs their ability to accurately copy the state $\rho_\beta$. Let us denote these partitions as $\mathcal{P}$, where $\cup^{k-1}_{i = 0} \mathcal{P}_i = \{E_d\}, \, \cap^{k-1}_{i = 0} \mathcal{P}_i = \emptyset $ and $\{E_d\}$ represents the energy levels of $H_{\ms S}$. We denote the sums of populations of $\rho_\beta$ with respect to these partitions as 
\begin{gather}
   \overline{p}_i:= \sum_j \bra{j} \rho_\beta \ket{j} \, : \, \ket{j} \in \mathcal{P}_i.
\end{gather}
With access to a $k$-level probe in a pure state $\ket{r}$, the agent can then apply the unitary
\begin{align}
    \hspace{-0.2cm}U_{\text{IE}} &\!=\! \sum_{i_0 \in \mathcal{P}_0} \! \ketbra{i_0}{i_0}\! \otimes\! S_{r,0} + \!\dots\!   \nonumber \\ &\hspace{2cm}\!\dots\!+\hspace{-0.5cm}\sum_{i_{k-1} \in \mathcal{P}_{k-1}} \hspace{-0.5cm}\ketbra{i_{k-1}}{i_{k-1}} \!\otimes\! S_{r,k-1}+ \mathbb{1}_{\text{Rest}}\nonumber \\
    &= \sum^{k-1}_{j = 0} \ketbra{i_j}{i_j} \otimes S_{r,j} + \mathbb{1}_{\text{Rest}} \, : \, \forall \, \ket{i_j} \in \mathcal{P}_j ,\label{eq:coupling}
\end{align}
which produces the probe state
\begin{align}
    \sigma = \tr\{U_{\text{IE}}(\rho_\beta \otimes \ketbra{r}{r})U_{\text{IE}}^\dagger\} = (\bar{p}_0,  \bar{p}_1, \dots  \bar{p}_{k-1})
\end{align}
where the agent has acquired a state which involves $k$ partial sums of the $d$ populations making up $\rho_\beta$. Note that the coupling $U_{\text{IE}}$ as defined in Eq.~\eqref{eq:coupling}, does not perturb the unknown state of the system, allowing the agent to couple multiple probes sequentially to the system. By choosing different controlled rotations in each round, the agent can obtain different information about the system without disturbing it.

\paragraph*{Different Partitions \& Different Measurements.}  With access to a given $k$-level probe, an agent may choose to couple it with the $d$-level system they wish to learn in different ways, corresponding to various levels of control. The number of possible couplings can be determined by counting the number of $k$-tuples, where each entry in the tuple can take one of $d$ possible symbols, with each symbol representing a specific partition. However, we exclude cases where all entries are filled with the same symbol, of which there are $k$, as well as account for the permutations. Consequently, the total number of distinct couplings is given by $\frac{k^{d}-k}{k!} = \frac{k^{d-1}-1}{(k-1)!}$.

We may ask how many couplings result in partitions of the same size, where the energy levels of the system are coupled non-equivalently to the energy levels of the probe. These couplings represent different measurement outcomes with the same level of control. We refer to these couplings with partitions of the same size as a \textit{measurement setting}, which we represent by the vector $\vec{t} = [t_0, t_1, \ldots, t_{k-1}]$. Here $t_i$ denotes the number of times the symbol $i$ appears in the $k$-tuple and so the number of energy levels in the $i$th partition, $|\mathcal{P}_i| = t_i$. Note that the sum of the elements of $\vec{t}$ must always equal $d$, $\sum^{k-1}_{i = 0} t_i = d$ to correspond to the $d$ populations of $\ms S$. Therefore the no. of vectors $\vec{t}$ corresponds to the number of possible $k$-tuples of a given \textit{type} i.e. those that sum to $d$. In this way the number of \textit{measurement outcomes, probes, and estimates} for a given \textit{measurement setting} is
\begin{equation}
    m:= \frac{1}{r!}\binom{d}{\vec t}=\frac{1}{r!}\binom{d}{t_0, t_1, \dots t_{k-1}}
\end{equation}
where $r$ is the number of repetitions in $\vec t$. We will denote the set of possible estimate $\omega_i$ for a given measurement setting using the set $\Omega$.

To make this more concrete, consider an agent attempting to learn the thermal state of a 4-level system with diagonal elements $(p_0, p_1, p_2, p_3)$ using a qubit probe. There are two possible measurement settings. The first setting corresponds to $\binom{4}{1,3}$, yielding the probe states $(p_0, p_1 + p_2 + p_3)$, $(p_1, p_0 + p_2 + p_3)$, $(p_2, p_0 + p_1 + p_3)$, and $(p_3, p_0 + p_1 + p_2)$, which represent the four possible measurement outcomes for $d = 4$ and $\vec{t} = [1,3]$. The second setting, corresponding to less control, mixes more energy levels together, resulting in three possible measurement outcomes with the probe states $(p_0 + p_1, p_2 + p_3)$, $(p_0 + p_2, p_1 + p_3)$, and $(p_0 + p_3, p_1 + p_2)$, for $d = 4$ and $\vec{t} = [2,2]$.

Finally, the information extraction protocol, establishes a way to transfer information from an unknown system to a lower-dimensional probe system in a manner that is ignorant of the unknown state, using controlled-SWAP unitaries. Let us now investigate the task of spreading this information throughout a higher-dimensional memory system without creating coherences in the memory, to be faithful to the initial thermal state.

\subsection{Estimate Generation}
\label{sec:est_gen_protocol}
We will call the process of taking the information from a $k$-dimensional measurement probe and promoting it to a $d$-dimensional system in a way that spreads the information equally with respect to the priors given by the chosen coupling, \textit{estimate generation}. Here the agent takes the $k$ level measurement probe $\sigma$ which now contains information about the populations of $\rho_\beta$, and unitarily spreads that information across a $d$ level \textit{memory} system in a way which respects \textit{Jaynes’ principle}~\cite{jaynes1957informationI}. To begin with, we will consider the memory to be in an energy eigenstate of $H_{\ms S}$) . That is, if the state of the measurement probe is $\sigma = (\overline{p}_0,\overline{p}_1,\dots, \overline{p}_{k-1})$ then unitarily correlating it w.l.o.g. with a memory in state $\ket{0}_d$, we expect to obtain a marginal state for the memory of the form \begin{gather}
\hspace{-0.3cm}\omega \!=\! \underbrace{\left(\colorbox{blue!30}{$\frac{\overline{p}_0}{t_0},\dots, \frac{\overline{p}_0}{t_0}$}\right.}_{t_0},
\underbrace{\colorbox{red!30}{$\frac{\overline{p}_1}{t_1},\dots \frac{\overline{p}_1}{t_1}$}}_{t_1},\dots, \underbrace{\left.\colorbox{orange!30}{$\frac{\overline{p}_{k-1}}{t_{k-1}},\dots,\frac{\overline{p}_{k-1}}{t_{k-1}}$}\right)}_{t_{k-1}}. \label{eq:omega}\end{gather}
This is a $d$ level diagonal state in the energy eigenbasis of $H_{\ms S}$, which we call an \textit{estimate}. Recall that $\overline{p}_i$ is the sum of the populations of $\rho_\beta$ corresponding to the partition $\mathcal{P}_i$, which contains $t_i$ levels. Therefore, we can see that $\omega$ is composed of blocks of size $t_i$, with each entry in the block equal to the average population of $\mathcal{P}_i$, given by $\frac{\overline{p}_i}{t_i}$. 

\section{Estimate Manipulation \label{S:estimate-manipulation}}

\begin{figure*}[t]
        \centering
	\includegraphics{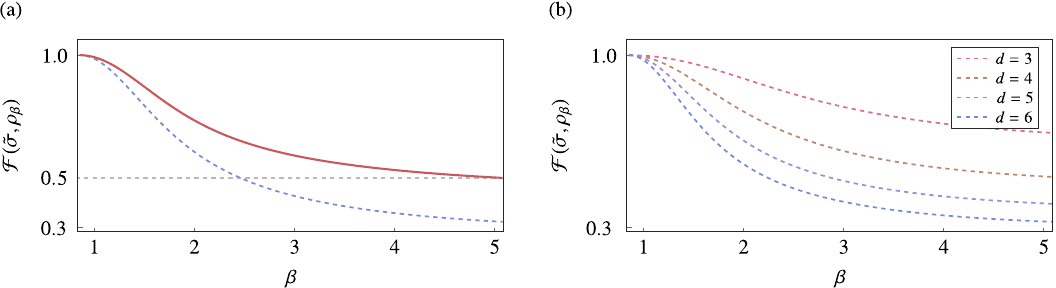}
	\caption{\justifying{ 
 \emph{Fidelity \& measurement settings}. For a thermal state $\rho_{\beta}$ described by an equidistant energy spectrum, we plot the fidelity as a function of $\beta$ for (a) two measurement settings $\vec{t} = [1,2,3]$ (red solid curve) using qutrit probes and $\vec{t}_2 = [1,5]$ (dashed blue curve) using qubit probes. We observe that the former is superior as the higher dimensionality allows the agent to couple less levels of the unknown system to more levels of the probe system. Focusing on qubit probes (b) for the general measurement $\vec{t}  = [1, d - 1]$, we notice that the fidelity decreases with the dimension of the estimated thermal state. In this case, the information is spread over the entries of the estimate, making it challenging to reconstruct $\rho_{\beta}$.}}
\label{F-fidelity}
\end{figure*}

In a given measurement setting $\vec{t}$, which describes how the agent's $k$-level probes can couple to the $d$-dimensional system $\rho_\beta$ to extract information, the agent can obtain $m$ estimate states $\omega$, each containing coarse-grained information about $\rho_\beta$. Here, we will examine how the agent can manipulate the set of estimate states they have generated to approximate the target state $\rho_\beta$ as closely as possible. Specifically, for a $d$-level system in the state $\rho_\beta$ which coupled to $k$-level probes, the agent generates a set of $m$ estimate states ${\omega}$ in quantum memories. We are interested in the fidelity between each estimate state and the target probes
\begin{gather}
    \mathcal{F}(\rho_\beta , \omega_i) \, : \,  \omega_i \in \Omega.
\end{gather}
And whether, by manipulating the $m$ estimates through a unitary process $V$, we can produce a state $\widetilde{\omega}$ with a fidelity
\begin{gather}
    \mathcal{F}(\rho_\beta , \widetilde{\omega}) \text{ where } \widetilde{\omega} = \tr_{m-1}\Biggl\{V \bigotimes_{ \omega_i \in \Omega} \omega_i V^\dagger\Biggl\},
\end{gather}
which is greater than that of the individually obtained estimates. For brevity, we will use the notation $\overline{\omega} = \bigotimes_{\omega_i \in \Omega} \omega_i$ to represent the joint state of all the estimates. 
\subsection{Symmetric -- Knowledge Distribution}
A scenario of particular interest is the symmetrisation of the $m$ estimated states, which allows an agent to spread the information contained in each estimate uniformly. This process results in $m$ identical processed states, denoted as $\tilde{\omega}_s$, which can then be used to approximate $\rho_{\beta}$ for a given task. To achieve this, we project the state $\overline \omega$ onto the symmetric subspace using the symmetrisation projector $\Pi_n = \frac{1}{n!}\sum_{\pi \in \mathcal{S}_n} R_\pi$, with
\begin{align}
    R_\pi = \sum^{d}_{i_1,\dots,i_n = 1} \ketbra{i_{\pi(1)}\dots i_{\pi(n)}}{i_1 \dots i_n} ,
\end{align}
where $R_\pi$ is the unitary representation of the permutation $\pi \in \mathcal{S}_n$ and $\mathcal{S}_n$ is the group of permutations of $n$ objects.
\begin{proposition}
The symmetrised estimate state $\widetilde{\omega}_s$ is the mean over all $m$ estimate states $\Omega$,
$$\widetilde{\omega}_s = \frac{1}{m} \sum_{\omega_i \in \Omega} \omega_i.$$
\end{proposition}
\textit{Proof.} Let us consider the symmetrisation of $m$ estimates 
\begin{gather}
   \widetilde{\omega}_s = \tr_{m-1}\{\Pi_m \bigotimes_{ \omega_i \in \Omega} \omega_i \Pi_m\}
\end{gather}
where the trace can be taken over any of the $m-1$ states to obtain the same reduced state, owing to the symmetry of $\Pi_m$. Now, we fix one site
\begin{align}
\widetilde{\omega}_s&=\tr_{m-1}\{\frac{1}{m}\sum_{\omega_i \in \Omega} \omega_i \otimes \Pi_{m-1} \bigotimes_{\substack{ \omega_j \in \Omega \\ \omega_j \neq \omega_i}} \omega_j \Pi_{m-1}\} 
\intertext{taking the sum over its possible values and symmetrising over the remaining $m - 1$ estimates}
&=\frac{1}{m} \sum_{\omega_i \in \Omega} \omega_i \tr_{m-1}\{\Pi_{m-1} \bigotimes_{\substack{ \omega_j \in \Omega \\ \omega_j \neq \omega_i}} \omega_j \Pi_{m-1}\}
\intertext{which, when projected into the symmetric subspace, form a valid quantum state with trace one, resulting in}
&= \frac{1}{m} \sum_{\omega_i \in \Omega} \omega_i. \hspace{0.5cm} \Box \,\, 
\end{align}
In this way, through symmetrisation, we can take $m$ estimates with different fidelities and convert them into $m$ identical states $\widetilde{\omega}_s$, each having the same fidelity with $\rho_{\beta}$ which is greater or equal to the average fidelity of the individual estimates.

\begin{corollary}
    The fidelity of the symmetrised estimate state $\widetilde{\omega}_s$ to the \textit{true} state $\rho_\beta$ is greater or equal to the average of the fidelities of the individual estimate states $\omega$,
    $$\mathcal{F}(\rho_\beta, \widetilde{\omega}_s) \geq \frac{1}{m} \sum_{\omega_i \in \Omega} \mathcal{F}(\rho_\beta, \omega_i).$$
\end{corollary}
\textit{Proof.} This follows directly from the prior proposition and the concavity of the fidelity i.e. $\mathcal{F}(\sigma,\sum_{i} p_i \rho_i) \geq \sum_i p_i \mathcal{F}(\rho_i, \sigma)$.
\begin{align}
    \mathcal{F}(\rho_\beta, \widetilde{\omega}_s) &= \mathcal{F}\left(\rho_\beta, \sum_{\omega_i \in \Omega}\frac{1}{m} \omega_i\right) \\
    &\geq \frac{1}{m} \sum_{\omega_i \in \Omega}\mathcal{F}\left(\rho_\beta, \omega_i\right)  \hspace{0.5cm} \Box \,\,
\end{align}
 This raises the question: under what conditions is the fidelity of the symmetrised estimate greater than the maximum fidelity of the individual estimates? To explore this, let us consider some examples, which we can compare to numerically tractable estimate states.

Let us examine how closely an agent, with access only to qubit probes, can reconstruct the state of an unknown $d$-level system in a thermal state after multiple rounds of our protocol. It is instructive to consider the extremal cases of possible measurement settings. Specifically, when $\vec{t} = [1, d-1]$, one specific level is coupled to the ground state of the qubit, while the remaining levels are coupled to the excited state. In contrast, for $\vec{t} = [d/2, d/2]$, half of the levels are coupled to each of the two states of the qubit. Starting with the former case, we find that the corresponding estimate states take the form
\begin{equation}
    \!\!\!\!\omega_0 \!=\!  \!\begin{bmatrix}
p_0 \\
\frac{1-p_0}{d-1}
 \\
 \vdots
 \\
 \frac{1-p_0}{d-1}
\end{bmatrix} \,  ,  \, \omega_1 \!=\! \! \begin{bmatrix}
\frac{1-p_1}{d-1} \\
p_1
 \\
 \vdots
 \\
\frac{1-p_1}{d-1}
\end{bmatrix}, \, \,  \!\hdots \!\, \, , \omega_{d-1} \!=\! \! \begin{bmatrix}
\frac{1-p_{d-1}}{d-1} \\
\frac{1-p_{d-1}}{d-1}
 \\
 \vdots
 \\
p_{d-1}
\end{bmatrix}.
\end{equation}
Since the symmetrised estimate state $\tilde{\omega}_{[1,d-1]}$ is the mean over all $d-1$ estimate states, we find that
\begin{align}
\tilde{\omega}_{[1,d-1]} &= \frac{1}{d}(\omega_0 + ... + \omega_{d-1}) \nonumber \\
&= \frac{1}{d} \qty{\begin{bmatrix}
p_0 \\
p_1
 \\
 \vdots
 \\
p_{d-1}
\end{bmatrix}+ \frac{1}{d-1}\begin{bmatrix}
p_0 + (d-2) \\
p_1 + (d-2)
 \\
 \vdots
 \\
p_{d-1} + (d-2)
\end{bmatrix}} \nonumber \\
&= \frac{1}{d-1}\qty[\rho_{\beta}+\frac{(d-2)}{d}\mathbb{1}_d]. \label{Eq:symmetrised-estimate-1d}
\end{align}

From Eq.~\eqref{Eq:symmetrised-estimate-1d}, one can easily compute its fidelity $\mathcal{F}(\rho,\sigma) = \qty[\tr(\sqrt{\sqrt{\rho}\sigma\sqrt{\rho}})]^2$ with respect to the target thermal state $\rho_{\beta}$,
\begin{align}
    \!\!\!\mathcal{F}(\rho_\beta,\widetilde{\omega}_{[1,d-1]}) \!=\! \frac{1}{d-1}\left(\sum^{d-1}_{i = 0}\sqrt{p_i\left(p_i + \frac{d-2}{d}\right)}\right)^2 .
\end{align}
Note that in the limiting cases, where $\beta =0$, which implies $\rho_{\beta}\to \frac{\mathbb{1}_d}{d}$, the fidelity becomes $\mathcal{F}(\tilde{\omega},\rho_{\beta}) = 1$; while when $\beta \to \infty$, the fidelity becomes $\frac{2}{d}$. Moving on to the latter coupling, we find in Appendix~\hyperref[sec:d2_d2]{C-1}. that the estimates produced by this coupling when symmetrised will have entries 
\begin{align}
    \bra{i}\widetilde{\omega}_{[\frac{d}{2},\frac{d}{2}]}\ket{i} &= \hspace{-1cm}
        \sum^{d - 1}_{1 \leq \alpha_1 \leq \dots \leq \alpha_{d/2 - 1} \leq d - 1} \hspace{-1cm} p_i + p_{\alpha_1} + \dots + p_{\alpha_{d/2 - 1}}.
\end{align}
The $i$th entry of $\widetilde{\omega}_{[\frac{d}{2},\frac{d}{2}]}$ is a sum of sums involving $p_i$ and a different selection of $d/2 - 1$ populations $p_{\alpha_k}$, where each $p_{\alpha_k}$ cannot be equal to $p_i$. Each individual sum originates from a different estimate. By examining the combinatorics of this sum, we observe that $p_i$, which appears in all contributions, occurs $\binom{d - 1}{d/2 - 1}$ times, while any other population $p_l$ appears $\binom{d - 2}{d/2 - 2}$ times, giving
\begin{align}
 \bra{i}\widetilde{\omega}_{[\frac{d}{2},\frac{d}{2}]}\ket{i} \!=\! \frac{2}{md}\! \left[\! \binom{d-1}{d/2 -1}p_i \!+\! \binom{d-2}{d/2 - 2}\left(\sum^{d-1}_{j =1}p_j\right) \!\right]. \nonumber
\end{align}    
Applying Pascal's identity then gives the form of the symmetrised estimate state, we obtain
\begin{align}
\hspace{-0.3cm}\widetilde{\omega}_{[\frac{d}{2},\frac{d}{2}]} =  \frac{2}{\binom{d-1}{d/2 -1}}\left[ \frac{1}{d}\binom{d-2}{d/2 - 1}\rho_\beta + \binom{d-2}{d/2 - 2}\frac{\mathbb 1}{d}\right]. \label{eq:d2d2}
\end{align}
In Fig.~\ref{F-fidelity}, we plot the fidelity of a symmetrised estimate state with the unknown thermal state for different couplings and dimensions. In panel Fig.~\hyperref[F-fidelity]{\ref{F-fidelity}a}, for a 6-level unknown thermal state, we observe that an agent with access to a qutrit probe outperforms an agent with access to a qubit probe across all temperatures. Additionally, hotter states tend to be easier to reconstruct through symmetrisation. In panel Fig.~\hyperref[F-fidelity]{\ref{F-fidelity}b}, we examine how the performance of a fixed coupling type varies with increasing dimension of the unknown system, noting that as the dimension of the unknown system increases, it becomes more difficult to achieve a faithful reconstruction.

It is interesting to note that symmetrised estimate states bear a striking resemblance to optimal clones in the context of unitary operations that take $n$ copies of a pure state $\ket{\psi}$ and $m$ copies of the $\ket{0}$ state to form $m + n$ clones that are as close as possible to the original state. In these works, it was found that symmetrisation maps function as optimal cloners~\cite{gisin_massar_97,werner_98, scarani2005quantum} and result in clones of the form $\rho = \gamma \ketbra{\psi}{\psi} + (1-\gamma)\frac{\mathbb 1}{d}$ where $\gamma$ is the so-called Shrinking or Black-Cow factor~\cite{werner_98,bruss_ekert_98}. Optimal quantum cloning has a rich relationship with state estimation~\cite{bruss_ekert_98,Bruss_1999,acin_06}, so it is perhaps not too surprising that our investigation into estimating thermal states has revealed this link. Indeed, our symmetrisation protocol can be viewed as a form of quantum cloning, where each copy has lost some information in a consistent yet distinct way. Specifically, each copy undergoes a coarse-graining of the same type, but applied to different parts of its spectrum. 

\subsection{Asymmetric -- Knowledge concentration}
\label{sec:know_conc}
Given a $d$-dimensional system and a $k$-dimensional pure state probe, coupled according to the measurement setting $\vec{t}$, the agent obtains $m$ estimate states through the correlation of $m$ probes with the system, as described in Sec.~\ref{sec:est_gen_protocol}. While the preceding subsection focused on symmetrically spreading the acquired information across all $m$ estimates to improve them uniformly, we now turn to the asymmetric case. The question we now pose is: \textit{How close can one bring a single estimate to the target state by applying a global operation on the $m$ estimates, sacrificing $m - 1$ of them in the process?} Specifically, does there exist an operation $V$ such that 
\begin{gather}
    \widetilde{\omega} = \tr_{1,\dots,m-1}\{V \bigotimes_{ \omega_i \in \Omega} \omega_i V^\dagger\} \stackrel{?}{=} \rho_\beta.
\end{gather}
We dub this the task of \textit{knowledge concentration}.

\begin{figure}[t]
    \centering
    \includegraphics[width = \linewidth]{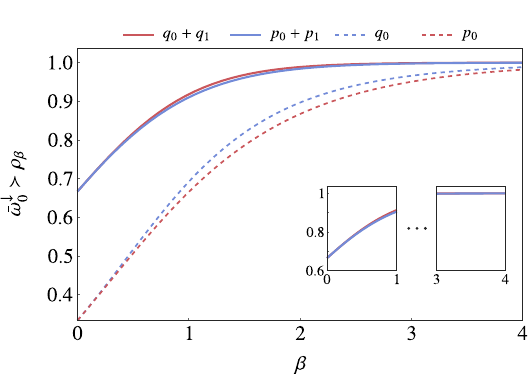}
    \caption{\emph{Majorisation inequalities}. For a three-level thermal state, we plot the marginal majorisation inequalities described above were $p_i$ are the populations of $\rho_\beta$ and $q_i$ are the eigenvalues of $\overline{\omega}^\downarrow_0$ for the estimates given in Eq.~\eqref{eq:qutrit_estimate} in Appendix~\hyperref[sec:warm_up]{D-1}, as a function of $\beta$.}
    \label{F-majorisation-inequalities}
\end{figure}

\begin{theorem}
Knowledge concentration is possible if and only if the first marginal of the ordered product of the estimate states majorises the target thermal state. That is, 
$$ \exists \, V \, : \,  \tr_{1, \dots, m-1}\{V \overline{\omega} V^\dagger\} = \rho_\beta \, \iff \,  \overline{\omega}^\downarrow_0 \succ \rho_\beta.$$
\end{theorem}

\emph{Proof Sketch.} For the necessary implication $(\!\!\implies\!\!)$, let us assume that such a $V$ exists and denote $V\overline{\omega}V^\dagger = \mu$. Then, by the Schur-Horn theorem~\cite{ marshall1979inequalities}, the majorisation relation $\diag{\overline{\omega}} \succ \diag{\mu}$ holds in the basis of $\omega$, i.e.
\begin{align}
    \sum^{s}_{n = 1} r^\downarrow_n \geq  \, \sum^{s}_{n = 1} \mu^\downarrow_n \hspace{1cm} \forall s \in [1, d^m] ,
\end{align}
where $r^{\downarrow}_n$ are the entries of $\diag{\omega}$ arranged in descending order, and $\mu_n$ are the entries of $\diag{\mu}$ in descending order. By construction, the first marginal of $\overline{\omega}$ has entries given by $\sum^{d^{m-1} j}_{n=1} r_n$, and $\mu_0 = \rho_\beta$ can also be expressed as $\sum^{d^{m-1} j}_{n=1} \mu_n = p_j$. Since $\rho_\beta$ is a thermal state, the $p_j$ are arranged in descending order. The global majorisation then implies $\sum^{d^{n-1}}_{n=1} r^\downarrow_n \geq \mu^\downarrow_n = p_j$ and thus $\overline{\omega}_{0}^\downarrow \succ \rho_\beta$, as required.

For the sufficient statement $(\!\!\impliedby\!\!)$, assume that $\overline{\omega}^\downarrow_0 \succ \rho_\beta$, and take a constructive approach. Since $\rho_\beta$ is diagonal, the operator $V$ cannot introduce coherences in $V\overline{\omega}V^\dagger$ such that $\tr_{i_2,\dots,i_{m-1}}\{V\overline{\omega}V^\dagger\}$ is non-diagonal. Therefore, let us consider a $V$ that is composed of a direct sum of unitaries acting on subspaces, where the coherences do not affect the first marginal. One possible choice for such a decomposition could be
\begin{gather}
  \hspace{-0.3cm}  \mathcal{H}_{i_1,\dots,i_{m-1}} \!=\! \text{span}\left\{\ket{j \, j+i_1 \, j+i_2 \dots j+i_{m-1} }\right\}^{d-1}_{j = 0}.
\end{gather}
Here, the addition is modulo $d$, giving $\mathcal{H} = \bigoplus^{d-1}_{i_1,\dots,i_{m-1} = 0} \mathcal{H}_{i_1,\dots,i_{m-1}}$. In a given subspace $\mathcal{H}_{i_1,\dots,i_{m-1}}$, the basis kets never repeat and are always separated by at least one unit in Hamming weight in terms of their indices, ensuring that their coherences do not contribute to the first marginal. Constructing $V = \bigoplus^{d-1}_{i_1,\dots,i{m-1} = 0} U_{i_1,\dots,i_{m-1}}$ and decomposing $\overline{\omega}$ into blocks $\overline{\omega}_{i_1,\dots,i_{m-1}}$, we observe that, since $\overline{\omega}$ is diagonal, the action of one of these subspace unitaries $U_{i_1,\dots,i_{m-1}}$ on the block $\overline{\omega}_{i_1,\dots,i_{m-1}}$ induces a unitary stochastic matrix on the diagonal of this block, as
\begin{align}
    U_{i_1,\dots,i_{m-1}} \overline{\omega}_{i_1,\dots,i_{m-1}} &U_{i_1,\dots,i_{m-1}}^\dagger \nonumber\\ &=\ M_{i_1,\dots,i_{m-1}} \vec{r}_{i_1,\dots,i_{m-1}},
\end{align}
where $\vec{r}_{i_1,\dots,i_{m-1}}$ is the diagonal of $\overline{\omega}$ in this block and $M = |u^{i_1,\dots,i_{m-1}}_{j,k}|^2 \ketbra{j}{k}$ since $U_{i_1,\dots,i_{m-1}} = u^{i_1,\dots,i_{m-1}}_{j,k} \ketbra{j}{k}$. Focusing on the case where all $M_{i_1,\dots,i_{m-1}} = M$, we have $M\sum_{i_1,\dots,i_{m-1}} \vec{r}_{i_1,\dots,i_{m-1}} = M \diag{\overline{\omega}_0}$ from which we can reach any $\rho$ s.t. $\diag{\overline{\omega}_0} \succ \diag \rho$ by the Schur-Horn Theorem~\cite{marshall1979inequalities}. Since we assume $\overline{\omega}_0^\downarrow \succ \rho_\beta$ as a premise then $\diag{\overline{\omega}_0} \succ \diag \rho_\beta$ completing the proof. $\Box$

A more detailed presentation of this proof is provided in Appendix~\hyperref[sec:know_conc_proof]{D} including a warm up qutrit example in Appendix~\hyperref[sec:warm_up]{D-1}. Finally, in Fig.~\ref{F-majorisation-inequalities}, we plot the majorisation inequalities for $\rho_\beta$ (with probabilities $p_i$), representing an unknown qutrit thermal state of a system with an equidistant spectrum, and $\overline{\omega}$ (with probabilities $q_i$), obtained from three estimates as described in Appendix~\hyperref[sec:warm_up]{D-1}. Interestingly, we observe that for this specific dimension and Hamiltonian, exact recovery through knowledge concentration is always possible. However, for different Hamiltonians or dimensions, this may not necessarily be the case.

\section{Toy Model -- Extracting Work from an extended Szilard Engine with Limited Memory}
\label{sec:work_extrac}
\begin{figure}[b]
\centering
\includegraphics[width=\linewidth]{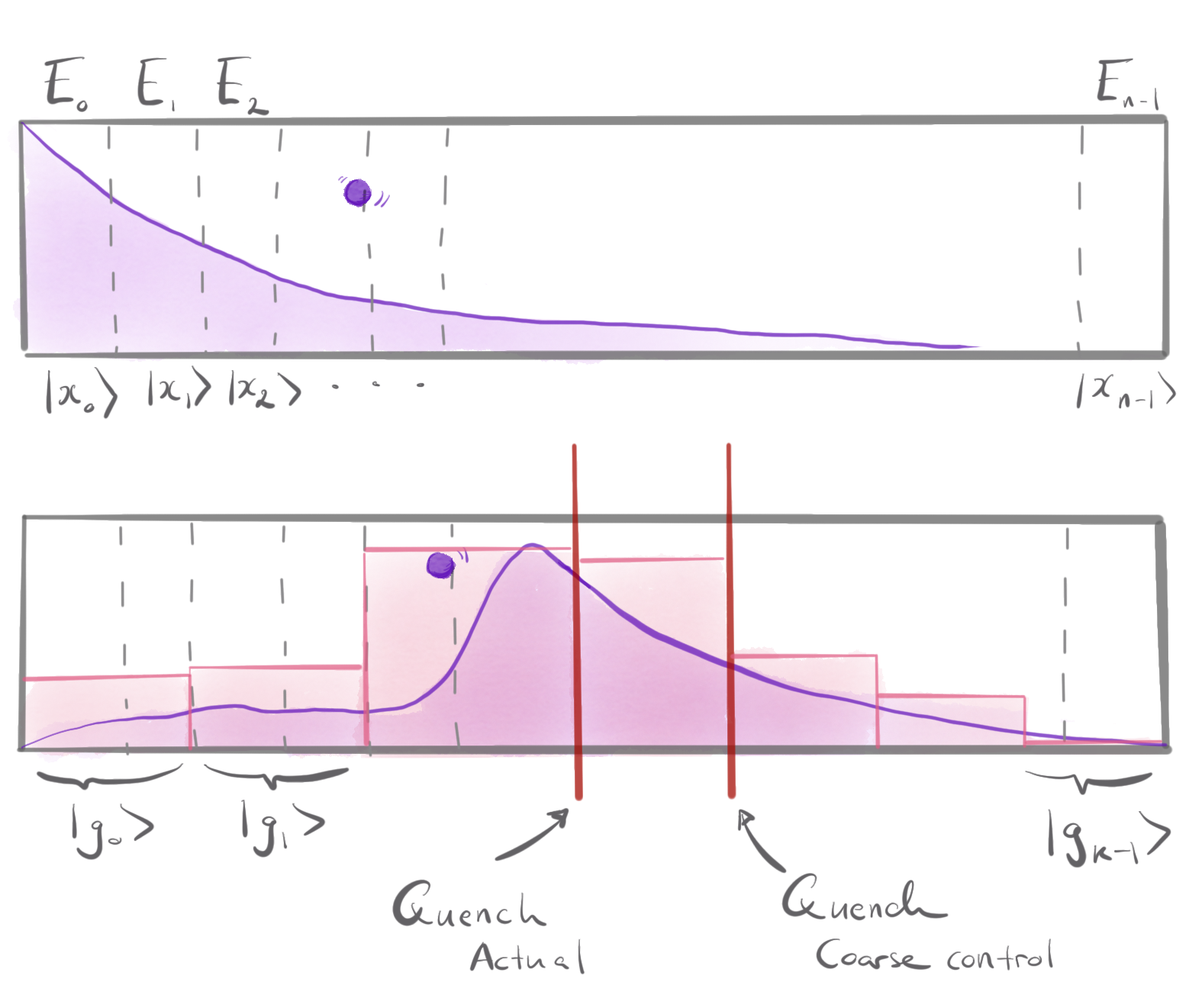}
    \caption{\emph{Extend Szilard engine with limited memory}. An illustration of the scenario considered in this section. An agent carries out a different quench changing the potential in a 1D box to extract work from a particle as it relaxes to a thermal state $\rho_\beta$. The quench carried out depends on the knowledge the agent has of the position distribution of the particle. A more coarse-grained understanding (pink) can lead to imperfect work extraction.}
\label{fig:szilard}
\end{figure}

To illustrate the applicability of the tools we have developed for handling coarse-grained control and limited probes, we now turn to a toy model. The Szilárd engine is a paradigmatic model in quantum thermodynamics, enabling a tractable analysis of Maxwell’s demon paradox through both theoretical and experimental approaches~\cite{szilard1929entropieverminderung,vedral_review,junior2025friendlyguideexorcisingmaxwells}. 

In this toy model, we consider a quantum particle trapped in a one-dimensional box with an energy gradient dependent on position. An agent is given a partition or piston that can be placed at any position along the box. By inserting the partition, the agent compresses the single-particle gas and can subsequently exploit its expansion to perform work. Whilst this work is seemingly carried out for free, it is equal to the work invested by an agent to learn the position of the particle and extract this work. But \textit{how does the size of the agent's memory and the precision of their control impact their ability to extract work} in this scenario? The position of the piston has not been considered in most erasure protocols but a recent work~\cite{cerisola_25} explored a scenario where erasing to $\ket{0}$ or $\ket{1}$ leads to different work output. In~\cite{cerisola_25} erasure of a bit encoded into a fermionic excitation in quantum dots led to different work extracted in this way due a chemical potential difference across the baths. The model presented here can be seen as a theoretical extension of this idea where multiple microstates in an erasure protocol differ in energy.

This question can be answered as an application of the framework and tools we have developed so far. Let us consider a simplified description of the particle’s evolution, modelling its position through a discretization of the Hamiltonian $H_{S} = \sum^{d-1}_{i = 0} E_i \ketbra{x_i}{x_i}$ where $\{\ket{x_i}\}^{d-1}_{i = 0 }$ denote the position eigenkets. In this discrete setting, the particle evolves under the translation operator $V_\text{Shift}(t) = e^{-iH_\text{Shift}t}$ generated by the Hamiltonian $H_{\text{Shift}} = \sum^{d-2}_{i = 0} \ketbra{x_i}{x_{i+1}} + \ketbra{x_{i+1}}{x_i}$. The work extraction protocol proceeds as follows: (i) the particle is initially in a thermal state $\rho_\beta$ of $H_S$ at inverse temperature $\beta$.  (ii) The particle is unitarily driven out of equilibrium by evolving under the shift operator for a time  $t_1$ giving $\rho' = V_\text{Shift}(t_1) \rho_\beta {V_\text{Shift}(t_1)}^\dagger$. (iii) The agent unitarily acquires information about the state of the particle via $U_{\textrm{IE}}\left(\rho' \otimes \sigma\right)U_{\textrm{IE}}^\dagger$, where $\sigma$ denotes the initial memory state. This produces an estimate $\omega$ of the particle’s state. (iv) Based on the estimate $\omega$ the agent performs an instantaneous quench of the system Hamiltonian.
\begin{align}
    &H_\text{Quench}(\omega) = \hat{\Delta}(\omega)H_S \text{ where}\nonumber\\
    &\hat{\Delta}(\omega) = \!\sum^{d-1}_{i = 0 } \bra{x_i}\omega \ket{x_i} E_\text{Max} \,
    \ketbra{x_i}{x_i}, \label{eq:protocol}
\end{align}
where $E_\text{Max} = \max_i \{E_i\}$ is the largest energy of $H_S$. This quench increases the energy in regions proportional to the population perceived by the agent after driving, thus performing work on the system, denoted $W_\text{in}$. (v) Finally, the agent allows the system to rethermalize with the environment, returning to the state $\rho_\beta$. In this process, the energy landscape is restored to its original form, and dissipative work $W_\text{out}$ is extracted.

In this protocol, the agent invests
\begin{gather}
    W_\text{In}(\omega) = \text{tr}\{(H_S - H_\text{Quench}(\omega))\rho'\},
\end{gather}
when performing the quench, and subsequently extracts
\begin{gather}
    W_\text{Out}(\omega) = \text{tr}\{H_\text{Quench}(\omega)\rho'\} - \text{tr}\{H_S \rho_\beta\},
\end{gather}
as the system relaxes back into equilibrium with the environment. Whilst the total work $W_\text{Tot} =  W_\text{In}(\omega) + W_\text{Out}(\omega)$ is independent of $\omega$ as $W_\text{Tot} = \text{tr}\{H_S(\rho' - \rho_\beta)\}$ the individual invested and extracted work for a given estimate is different. In particular, we are interested in the dissipative work extracted by an agent as the system relaxes to the thermal state $W_\text{Out}(\omega)$ for different estimates. It follows from the definition of the symmetrised estimate that the work extractable with access to a symmetrised estimate $\widetilde{\omega}_{\vec t}$ is the average work extractable over its corresponding set of estimates $\Omega_{\vec t}$,
\begin{align}
    W_\text{Out}(\widetilde{\omega}_{\vec t}) &= \text{tr}\{H_\text{Quench}(\widetilde{\omega}_{\vec t})\rho'\} - \text{tr}\{H_S \rho_\beta\} \nonumber\\
   &=  \frac{1}{m}\sum_{\omega_j \in \Omega_{\vec t}}\text{tr}\left\{\hat{\Delta}(\omega_{j})H_S\rho'\right\} - \text{tr}\{H_S \rho_\beta\} \nonumber\\
   &= \sum_{\omega_j \in \Omega_{\vec t}} \frac{W_\text{Out}(\omega_j)}{m}.
\end{align}
\begin{figure}[t]
    \centering
\includegraphics{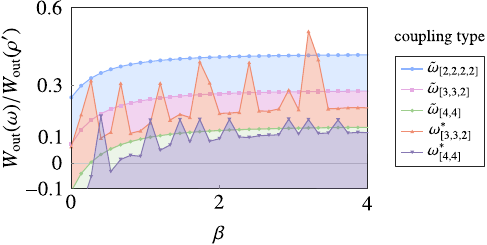}
\caption{\emph{Work Extraction by knowledge dependent quenches}. Dissipative work $W_\text{Out}(\omega)$ extracted by the agent during relaxation for different quenches $H_\text{Quench}(\omega)$ carried out on the system based on their estimate of the system $\omega$, as the initial temperature of the system decreases $\beta$. In this plot we assume the particle is trapped in a box with a potential well described by 8 positions $H_S = \sum^7_{i=0} E_i \ketbra{E_i}{E_i}$ with $E_i \in [0,7]$, the particle is driven under $H_\text{Shift}$ for $t_1 = 0.25$. The extracted work is shown as a ratio $W_\text{Out}(\omega)/W_\text{Out}(\rho')$ to the work extractable using a quench protocol with knowledge of the actual state $\rho'$. We plot random estimates $\omega^*_{\vec t}$ and symmetrised estimates $\widetilde{\omega}_{\vec t}$ of a given coupling type $\vec t$ where we see that as the dimension of the agent's probe increases, they are able to extract more work from the system. The code to obtain these numerics may be obtained at~\cite{code}.}
\label{fig:work_extract}
\end{figure}

The optimal work extraction protocol would have involved quenching the system such that the energies would be in descending order, leading to largest possible change when the system relaxes. But such a protocol presumes precise control of each energy level. As such, we chose the protocol Eq.~\eqref{eq:protocol} as it increases the energy of the system in a manner which is dependent on the estimate state $\omega$ obtained by the agent and their coarse control. That is, if the agent has access to a lower dimensional probe and is unable to distinguish some energy levels then they will quench these levels by the same amount under Eq.~\eqref{eq:protocol}. Their inability to distinguish energy levels and lack of knowledge impedes their ability to extract work. One can see this in Fig.~\ref{fig:work_extract} where in a scenario where the particle can take on 8 possible positions, knowledge obtained using a ququart probe leads to coarse graining which only allows for around half the possible work which can be extract via this protocol. As the dimension of the probe decreases, the more coarse graining the agent perceives and the less work they can extract.

\section{Conclusions \& Outlook \label{S:conclusions-outlook}}

The cost of reconstructing a quantum state or acquiring knowledge of it as an agent is an intricate and multifaceted subject. In this work, we shed light on a particular facet of this problem by framing knowledge acquisition as an agent's ability to reconstruct the diagonal of a quantum state in a fixed basis. In particular, we were interested in the role of limited control in this scenario and provided a two-step unitary model for knowledge acquisition in Sec.~\ref{S:information-extraction-estimate-generation} where the agent's ability to interact with the system of interest with access only to lower dimensional probes. In the information extraction phase, the agent correlated the probe with the system of interest and extracts coarse grained information of expectations of the system w.r.t a specific observable e.g. populations in the energy eigenbasis in the case of thermal states. In the estimate generation phase, the agent then attempt to unitarily reconstruct these populations in a memory system. Having constructed this model in Sec.~\ref{S:estimate-manipulation} we then consider what happens when an agent is allowed to interact with the system using their dimensionally restricted probes many times, obtaining several different estimate states. We explored two different strategies for manipulating the estimate states towards improving the agent's knowledge of the system (i) a symmetric approach where the information of the estimate states is averaged out over them (ii) an asymmetric approach where knowledge is concentrated into a single estimate allowing for complete recovery of the diagonal of the system if a majorisation condition is fulfilled. In Sec.~\ref{sec:work_extrac} we contextualised our results by considering a toy model of knowledge dependent work extraction. In particular we considered an agent quenching a particle in a 1D potential landscape based on the particle's position. The work the agent is able to extract in this setting is dependent on their knowledge of the position of the particle. In our model, the quality of this knowledge is dependent on the dimension of the probe the agent has access to and whether multiple rounds of interaction with the system were possible, allowing for the agent to apply symmetrisation. In Fig~\ref{fig:work_extract} we directly see these insights verified numerically as quench protocols performed better with access to higher dimensional probes and even better on average with access to symmetrised estimate states. 

Whilst the presented extended Szilard engine is only a toy model furnished by numerical results, the techniques we developed and exhibited using this model should allow for a better understanding of the resources agents need in quantum thermodynamics. The thermodynamic cost of information acquisition is a longstanding problem, dating back to Brillouin and his attempt to resolve the paradox of Maxwell's demon~\cite{brillouin1953negentropy}. More recently, fundamental bounds on the thermodynamic cost for information processing was obtained~\cite{PhysRevLett.102.250602}. However, a comprehensive framework that fully accounts for all the resources available to an agent is still lacking. Our basic insights show that knowledge acquisition is highly dependent on the resources an agent has access to, such as the dimension and purity or temperature of their probe system and memory. In this work, we have primarily focused on dimensional constraints, whereas previous studies have explored temperature restrictions with unrestricted probe dimensions~\cite{Guryanova_2020,minagawa2023universalvaliditysecondlaw}, or the transition from quantum to classical knowledge~\cite{Debarba2024}. To fully address the question of when the cost of knowledge acquisition is outweighed by the benefits of an agent's performance at a given task, one must combine our insights with those of~\cite{Guryanova_2020,Debarba2024,minagawa2023universalvaliditysecondlaw} to gain a complete understanding of the required resources.

Throughout the development of this work we studied other applications. The first given in Appendix~\ref{A:applications} was the impact of our techniques in an ergotropic setting where we investigated optimal work extraction from two thermal states at different temperatures. The second given in Appendix~\ref{A:resource_theory} was a preliminary investigation in how imperfect knowledge of the \textit{free states} of a resource theory could impact an agents ability to transform states using free operations. We present these works together with other supplementary work in the Technical Matter presented after the main text. 

\paragraph*{Acknowledgements} 
J.X. would like to acknowledge the Quantum Resources 2023 conference which inspired this project and in particular thank Felix Binder for his encouragement in pursuing this problem. J.X is indebted to Max Lock, Giulia Mazzolla and Benjamin Stratton for insightful conversations. In particular, Tony Short for pointing out the maximum rank increase that can occur in the estimate generation procedure. J.X., F.C., P.B. and M.H. acknowledge funding from the European Research Council (Consolidator grant `Cocoquest’ 101043705) and financial support from the Austrian Federal Ministry of Education, Science and Research via the Austrian Research Promotion Agency (FFG) through the project FO999914030 (MUSIQ) funded by the European Union – NextGenerationEU.

\providecommand{\noopsort}[1]{}\providecommand{\singleletter}[1]{#1}%
\clearpage
\appendix
\onecolumngrid
\begin{center}
    \large \bfseries Technical Matter
\end{center}
\vspace{1em}
\let\addcontentsline\oldaddcontentsline %

\begingroup
\parskip=0pt
\setcounter{tocdepth}{2}
\tableofcontents
\endgroup

\section{von Neumann Measurement Scheme\label{App:von-Neumann-measurement scheme}} 

In this appendix, we provide additional details on the von Neumann measurement scheme and illustrate it with a simple example. We begin by considering a quantum system $\ms{S}$ that interacts unitarily with an ancillary system $\ms{A}$. After this interaction, the ancillary system is projectively measured, and the resulting statistics are equivalent to those obtained from directly measuring the main system $\ms{S}$.

For a set of projection operators ${\Pi_i}$, the probability of observing outcome $x_i$ for a quantum state $\rho$ is given by the Born rule $p(x_i) = \tr\left\{\Pi_i \rho \right\}$. According to the von Neumann-L\"uders projection postulate~\cite{petruccione}, the post-measurement state corresponding to outcome $x_i$ is
\begin{gather}
    \rho_{x_i} = \frac{\Pi_i \rho \Pi_i}{\tr\left\{\Pi_i \rho \right\}}.
\end{gather}
This implies that the measurement, involving the full set of projectors, results in the mixed state:
\begin{gather}
\rho'=\sum_i p(x_i) \rho_{x_i} = \sum_i \Pi_i \rho \Pi_i .
\end{gather} 
In the von Neumann scheme, the measurement statistics can be encoded into an ancillary system $\ms{A}$, whose dimension matches the number of projectors in the measurement. This encoding is represented as
\begin{gather}
    \rho'_{\ms{AS}} = \sum_i p(x_i) \ketbra{a_i}{a_i} \otimes \rho(x_i)
\end{gather}
where $\ket{a_i}$ represents the states of the ancillary system. By projectively measuring the ancillary system, the measurement statistics of the system of interest can be recovered indirectly as:
\begin{gather}
\rho' = \sum_i p(x_i)\rho_{x_i}  = \tr_{\ms A}\left\{\rho'_{\ms{AS}}\right\}.
\end{gather}
As an example, consider measuring a qubit state $\ket{\psi} \in \mathbb{C}^{\otimes 2}$ in the $\sigma_z$ basis, with access to an ancillary qubit initially in the $\ket{0}_{\ms A}$ state. By unitarily acting on these qubits with
\begin{gather}
 V = \ketbra{0}{0}\otimes \mathbb{1} + \ketbra{0}{0}\otimes \sigma_z    
\end{gather}
gives the state $V(\ket{\psi}\otimes \ket{0}_A)= \braket{0|\psi}\ket{0}\otimes\ket{0}_A + \braket{1|\psi}\ket{1}\otimes\ket{1}_A$, which effectively encodes the statistics of the qubit of interest into the statistics of the ancillary qubit.

\section{Estimate Generation}
To generate the estimate, the agent begins with the state $\sigma \otimes \ketbra{0}_d$, where the non-zero elements are given by $\overline{p}_i \ket{i,0}$. The goal is to apply a unitary transformation to the initial state of the probe and the \textit{empty} memory, producing a global state in which the marginal state of the memory corresponds to the desired estimate state
\begin{gather}\omega = \tr_{\ms S}\{U(\sigma \otimes \ketbra{0}{0}_d)U^\dagger\}.\end{gather}
Considering the structure of $\omega$ in Eq.~\eqref{eq:omega}, the challenge is twofold (i) we need the elements to appear in the correct positions within the marginal. Specifically, we want each $p_i$ to be distributed across $t_i$ entries in the marginal, starting at position $\alpha_i = \sum^{i-1}_{j=0} t_j$, with $\alpha_0 = 0$, and ending at $\alpha_i + t_i$. (ii) We must ensure that the partial trace of the global state $U(\sigma \otimes \ketbra{0}{0}_d)U^\dagger$ results in $\omega$. Thus, when carrying out operations that create off-diagonal elements, we must make sure they populate entries in the matrix which do not contribute to the partial trace.

To achieve (i) we begin by permuting elements on the diagonal of the global state, using a two-level swap operation $S_{l,k}$, which does not induce coherences. In particular we carry out
\begin{align}
    \overline{p}_i \ket{i,0} \xrightarrow[]{S_{(i,0),(0,\alpha_i)}} \overline{p}_i \ket{0,\alpha_i}.
\end{align}
This operation rotates all the information into a block of the global state corresponding to the probe $\sigma$ being in $\ket{0}$, spanned by $\{\ket{0,n}\}^{d}_{n = 0}$. This means we only need to consider this block for the partial trace. Next, we apply a quantum Fourier transform, which evenly distributes $p_i$ across the desired subspace. This subspace must be spanned by $t_i$ basis elements to ensure that the correct number of entries appear in the marginal. We choose the subspace with basis $\{\ket{j,\alpha_i + j}\}^{t_i -1}_{j = 0}$, ensuring that any coherences generated within this subspace do not contribute to the partial trace. Performing the following operation
\begin{gather}
    \overline{p}_i \ket{0,\alpha_i} \xrightarrow[]{QFT_{t_i}} \sum^{t_i - 1}_{j = 0}\frac{\overline{p}_i}{\sqrt{t_i}} e^{i 2\pi j/t_i} \ket{j, \alpha_i + j}.
\end{gather}
creates $k$ blocks of size $t_i$ within the the global state's $\ket{0,n}$ subspace. This transformation is given by the following unitary:
\begin{equation}
    U_i^F =\sum_{j,l=0}^{t_i-1}\frac{e^{\frac{i2\pi jl}{t_i}}}{\sqrt{t_i}}\ketbra{j,\alpha_i + j}{l,\alpha_i+l}\oplus\mathbb{1}_{\text{Rest}}.
\end{equation}
Taking the partial trace, we need to consider the states
\begin{gather}
\sum^{k - 1}_{i = 0} \sum^{t_i - 1}_{j,m = 0}\frac{\overline{p}_i}{t_i} e^{\frac{i 2\pi(j-m)}{t_i}} \ketbra{j, \alpha_i + j}{m, \alpha_i + m},
\end{gather}
which summing over the indices of the probe gives 
\begin{gather}
    \omega = \sum^{k-1}_{i = 0} \sum^{t_i -1}_{k = 0} \frac{\overline{p}_i}{t_i} \ketbra{\alpha_i + k}{\alpha_i +k },
\end{gather}
as desired. Note that in constructing the subspace ${\ket{j,\alpha_i + j}}^{t_i - 1}_{j=0}$, this method presumes that the dimension of the probe $\sigma$ is larger than the number of elements in the largest partition, i.e.,, $k \geq \max{t_i}$. See Appendix~\ref{App:two-examples} for two explicit examples of estimate generation, considering two different measurement settings.

\subsection{Which Estimates should an agent generate?}

\begin{figure}[t]
    \centering
    \includegraphics[width = 0.35\textwidth]{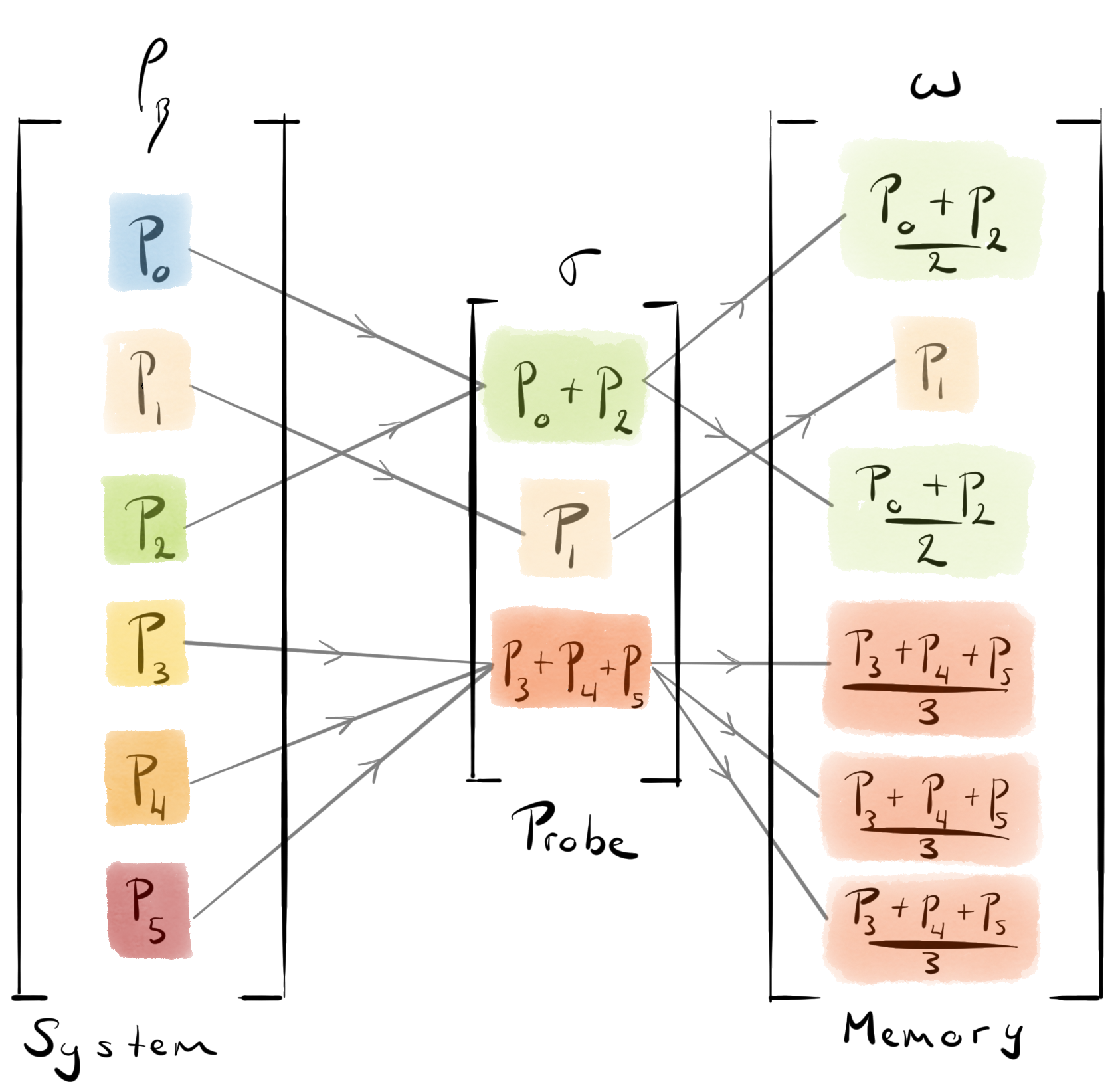}
    \caption{\emph{System, probe \& memory coupling}. An illustration showing the coupling that the agent carries out for a given measurement setting $\vec{t}$. If the agent couples $\ket{i}$ of the unknown system to $\ket{j}$ of the probe, then the probe is coupled to the $\ket{i}$th level of the memory. This is visualized for $d = 6$, $k = 3$, and $\vec{t} = (1, 2, 3)$.}
    \label{fig:coupling}
\end{figure}

Consider the situation where an agent is attempting to learn a $d = 6$ system using a $k = 3$ probe and measurement setting $\vec{t} = (1, 2, 3)$. They obtain $\sigma_{(2,1,3)} = (p_0 + p_2, p_1, p_3 + p_4 + p_5)$, as shown in Fig.~\ref{fig:coupling}. How will they decide to distribute the knowledge they have acquired in the memory system? Two possible estimate states for $\sigma_{(2,1,3)}$ are, 
\begin{gather}
\omega_{(2,1,3)} = \begin{pmatrix} \frac{p_0 + p_2}{2} \\ p_1 \\ \frac{p_0 + p_2}{2} \\ \frac{p_3 + p_4 + p_5}{3} \\ \frac{p_3 + p_4 + p_5}{3} \\ \frac{p_3 + p_4 + p_5}{3}\end{pmatrix} \:,\: \omega'_{(2,1,3)} = \begin{pmatrix}
\frac{p_3 + p_4 + p_5}{3} \\ \frac{p_0 + p_2}{2} \\ p_1 \\  \frac{p_0 + p_2}{2} \\ \frac{p_3 + p_4 + p_5}{3} \\ \frac{p_3 + p_4 + p_5}{3}     
\end{pmatrix},
\end{gather}
which, while equivalent up to a unitary, correspond to two quite physically distinct situations. For $\omega_{(2,1,3)}$, the agent correctly places the partial information $\overline{p}_i$ acquired about energy levels $\mathcal{P}i$ in the corresponding energy levels of the memory by carrying out the same coupling. In contrast, for $\omega'{(2,1,3)}$, the agent does not couple the probe to the memory in the same way that they coupled it to the system when extracting information. It is as if they have forgotten the initial coupling. For the remainder of this work, we will assume that the agent uses a coupling as described in the former situation, as visualized in Fig.~\ref{fig:coupling}, which we will also show leads to optimal fidelity when manipulating estimates symmetrically.

It is also important to note that we have assumed a one-to-one correspondence between probe and estimate states i.e., one $k$ level probe state leads to one $d$ level estimate state. In doing so, we are discarding situations where permutations of couplings between the probe and memory systems that lead to the same estimate. For a clear example, one may find the estimates for $d = 6, k = 2 $ and $\vec{t} = [3,3]$ and $d = 6, k=3, \vec{t} = [1,2,3]$ in Appendix~\hyperref[sec:d2_d2]{C-1}.
\subsection{Two Examples for estimate generation\label{App:two-examples}}

Here, we illustrate the estimate generation step with a non-trivial example. Consider a scenario where an agent seeks to reconstruct the state of a 6-dimensional thermal state $\rho_\beta = \left(p_0, p_1, p_2, p_3, p_4, p_5\right)$, and has access to a 3-dimensional probe, which is initially in the ground state, $\ket{0}_3$. The agent picks a measurement setting $\vec{t}_1 = \{1,2,3\}$ and extracts information from $\rho_\beta$ by applying a correlating unitary given by 
\begin{gather}
U_{\text{IE}} = \ketbra{0}{0}\otimes S_{0,0} + (\ketbra{1}{1} + \ketbra{2}{2} )\otimes S_{0,1} + (\ketbra{3}{3} + \ketbra{4}{4} + \ketbra{5}{5} )\otimes S_{0,2}.   
\end{gather}
After extraction, the agent's qutrit probe is in the state $\sigma = \tr_{\ms S}\left\{U_{\text{IE}}\left(\rho_\beta \otimes \ketbra{0}{0}_3 \right)U_{\text{IE}}^\dagger\right\} = (p_0, p_1 + p_2, p_3 + p_4 + p_5)$. Now, the goal is to distribute this information throughout a 6-dimensional memory. This is done by using the protocol described in Sec~\ref{sec:est_gen_protocol}. First, we coupled the qutrit probe with the memory
\begin{gather}
   \eta:= \sigma \otimes \ketbra{0}{0}_6 = p_0 \ketbra{00}{00} + (p_1 + p_2)\ketbra{10}{10} + (p_3 + p_4 + p_5) \ketbra{20}{20},
\end{gather}
and applies two-level SWAPs to arrange the populations in the correct positions within the $\ket{0i}$ top-left block of $\eta$. Specifically, the following two-level SWAP operations are applied: $S_{(1,0),(0,1)}$ and $S_{(2,0),(0,3)}$, to obtain
\begin{gather}
    (S_{20,03}S_{10,01})\eta (S^\dagger_{20,03} S_{10,01)})= p_0 \ketbra{00}{00} + (p_1 + p_2)\ketbra{01}{01} + (p_3 + p_4 + p_5) \ketbra{03}{03}.
\end{gather}
The next step consists of spreading the populations equally over the corresponding energy levels. This is achieved by applying a non-local Hadamard in the basis $\{\ket{01}, \ket{12}\}$, namely
\begin{gather}
    U^F_{01,12} = \frac{1}{\sqrt{2}}\left(\ketbra{01}{01} - \ketbra{12}{12} + \ketbra{01}{12} + \ketbra{12}{01} \right) \oplus \mathbb{1}_{\text{rest}},
\end{gather}
followed by quantum Fourier transforms over the $\ket{03}, \ket{14}, \ket{25}$ basis
\begin{align}
    U^F_{03,14,25} &= \frac{1}{\sqrt{3}}\left(\ketbra{03}{03} + \ketbra{03}{14} + \ketbra{03}{25} + \ketbra{14}{03}+ e^{i 2 \pi / 3} \ketbra{14}{14} + e^{ - i 2 \pi / 3}\ketbra{14}{25}\right. \nonumber \\ 
    &\hspace{5.5cm}+ \left. \ketbra{25}{03} + e^{ - i 2 \pi / 3} \ketbra{25}{14}+  e^{i 2 \pi / 3}\ketbra{25}{25} \right) \oplus  \mathbb{1}_{\text{rest}}.
\end{align}
Finally, by applying this sequence of unitaries, we obtain the estimated state. To recap, the entire protocol is represented by the product of unitary operations $\mathcal{U}_{123}:= S_{10,01}S_{20,03}U^F_{01,12}U^F_{03,14,25}$, where 
\begin{align}
    \omega &= \tr_{\ms S}\{\mathcal{U}_{123}(\sigma \otimes \ketbra{0}{0}_6)\mathcal{U}_{123}^\dagger \}\nonumber\\
        &= \left(
\begin{array}{cccccc}
 p_0 & 0 & 0 & 0 & 0 & 0 \\
 0 & \frac{1}{2} \left(p_1+p_2\right) & 0 & 0 & 0 & 0 \\
 0 & 0 & \frac{1}{2} \left(p_1+p_2\right) & 0 & 0 & 0 \\
 0 & 0 & 0 & \frac{1}{3} \left(p_3+p_4+p_5\right) & 0 & 0 \\
 0 & 0 & 0 & 0 & \frac{1}{3} \left(p_3+p_4+p_5\right) & 0 \\
 0 & 0 & 0 & 0 & 0 & \frac{1}{3} \left(p_3+p_4+p_5\right) \\
\end{array}
\right).
\end{align}

If instead the agent chose the coupling $\vec{t}_2 = (2,1,3)$ and obtains the state $\sigma_2 = (p_0 + p_1, p_2, p_3 + p_4 + p_5)$, then the series of unitaries are now $\mathcal{U}_{213} := S_{10,00}S_{20,03}U^F_{00,11}U^F_{03,14,25}$, which leads to
\begin{align}
    \omega &= \tr_{\ms S}\{\mathcal{U}_{213}(\sigma \otimes \ketbra{0}{0}_6)\mathcal{U}_{213}^\dagger \}\nonumber\\
        &= \left(
\begin{array}{cccccc}
 \frac{1}{2}(p_0+p_1) & 0 & 0 & 0 & 0 & 0 \\
 0 & \frac{1}{2}(p_0+p_1) & 0 & 0 & 0 & 0 \\
 0 & 0 & p_2 & 0 & 0 & 0 \\
 0 & 0 & 0 & \frac{1}{3} \left(p_3+p_4+p_5\right) & 0 & 0 \\
 0 & 0 & 0 & 0 & \frac{1}{3} \left(p_3+p_4+p_5\right) & 0 \\
 0 & 0 & 0 & 0 & 0 & \frac{1}{3} \left(p_3+p_4+p_5\right) \\
\end{array}
\right).
\end{align}
the estimate state corresponding to the measurement setting $\vec{t}_2$. Here, the protocol $\mathcal{U}_{213}$ differs to $\mathcal{U}_{123}$ only slightly since the position of the degenerate 2-dimensional subspace in the estimate state has changed.

\section{Symmetrisation \label{App:symmetrisation}}
\subsection{$d=6$, $k=2$, $\vec{t}=\{d/2,d/2\}$}
\label{sec:d2_d2}
To understand this situation, we first consider the estimate states for $d = 6$, $k = 2$, and $t = \{3,3\}$, where we have $\frac{1}{2} \binom{6}{3,3} = \frac{6!}{3!3!2} = 10$ estimates. Note that we divide the number of estimates by $2$ to avoid redundancy due to repetition in the $t$-vector. In other words, there would be another 10 estimates with equivalent populations to the 10 shown below, but these would be flipped across the $3$rd level.
\begin{align}
\begin{split}
\omega_0 &=  \begin{pmatrix}
\frac{p_0 + p_1 + p_2}{3}\\[0.5ex] 
\frac{p_0 + p_1 + p_2}{3}\\[0.5ex] 
\frac{p_0 + p_1 + p_2}{3}\\[0.5ex] 
\frac{p_3+p_4+p_5}{3} \\[0.5ex]
\frac{p_3+p_4+p_5}{3} \\[0.5ex]
\frac{p_3+p_4+p_5}{3}  
\end{pmatrix}  \:\: , \:\: \omega_1 = \begin{pmatrix}
\frac{p_0 + p_1 + p_3}{3}\\[0.5ex] 
\frac{p_0 + p_1 + p_3}{3}\\[0.5ex] 
\frac{p_2 + p_4 + p_5}{3}\\[0.5ex] 
\frac{p_0 + p_1 + p_3}{3} \\[0.5ex]
\frac{p_2 + p_4 + p_5}{3} \\[0.5ex]
\frac{p_2 + p_4 + p_5}{3} 
\end{pmatrix}
 \:\: , \:\: \omega_2 =   \begin{pmatrix}
\frac{p_0 + p_1 + p_4}{3}\\[0.5ex] 
\frac{p_0 + p_1 + p_4}{3}\\[0.5ex] 
\frac{p_2+ p_3 + p_5}{3}\\[0.5ex] 
\frac{p_2 + p_3 + p_5}{3} \\[0.5ex]
\frac{p_0 + p_1 + p_4}{3} \\[0.5ex]
\frac{p_2 + p_3 + p_5}{3}  
\end{pmatrix}
  \:\: , \:\: \omega_3 =   \begin{pmatrix}
\frac{p_0 + p_1 + p_5}{3}\\[0.5ex] 
\frac{p_0 + p_1 + p_5}{3}\\[0.5ex] 
\frac{p_2+ p_3 + p_4}{3}\\[0.5ex] 
\frac{p_2 + p_3 + p_4}{3} \\[0.5ex]
\frac{p_2 + p_3 + p_4}{3} \\[0.5ex]
\frac{p_0 + p_1 + p_5}{3}  
\end{pmatrix}  \:\: , \:\: \omega_4 = \begin{pmatrix}
\frac{p_0 + p_2 + p_3}{3}\\[0.5ex] 
\frac{p_1 + p_4 + p_5}{3}\\[0.5ex] 
\frac{p_0+ p_2 + p_3}{3}\\[0.5ex] 
\frac{p_0 + p_2 + p_3}{3} \\[0.5ex]
\frac{p_1 + p_4 + p_5}{3} \\[0.5ex]
\frac{p_1 + p_4 + p_5}{3}   
\end{pmatrix} \:\: , \\[0.5ex]
\omega_5 &= \begin{pmatrix}
\frac{p_0 + p_2 + p_4}{3}\\[0.5ex] 
\frac{p_1 + p_3 + p_5}{3}\\[0.5ex] 
\frac{p_0+ p_2 + p_4}{3}\\[0.5ex] 
\frac{p_1 + p_3 + p_5}{3} \\[0.5ex]
\frac{p_0 + p_2 + p_4}{3} \\[0.5ex]
\frac{p_1 + p_3 + p_5}{3}    
\end{pmatrix}  \:\: , \:\: \omega_6 = \begin{pmatrix}
\frac{p_0 + p_2 + p_5}{3}\\[0.5ex] 
\frac{p_1 + p_3 + p_4}{3}\\[0.5ex] 
\frac{p_0+ p_2 + p_5}{3}\\[0.5ex] 
\frac{p_1 + p_3 + p_5}{3} \\[0.5ex]
\frac{p_1 + p_3 + p_4}{3} \\[0.5ex]
\frac{p_0 + p_2 + p_5}{3}  
\end{pmatrix}  \:\: , \:\: \omega_7 =   \begin{pmatrix}
\frac{p_0 + p_3 + p_4}{3}\\[0.5ex] 
\frac{p_1 + p_2 + p_5}{3}\\[0.5ex] 
\frac{p_1+ p_2 + p_5}{3}\\[0.5ex] 
\frac{p_0 + p_3 + p_4}{3} \\[0.5ex]
\frac{p_0 + p_3 + p_4}{3} \\[0.5ex]
\frac{p_1 + p_2 + p_5}{3}   
\end{pmatrix}  \:\: , \:\: \omega_8= \begin{pmatrix}
\frac{p_0 + p_3 + p_5}{3}\\[0.5ex] 
\frac{p_1 + p_2 + p_4}{3}\\[0.5ex] 
\frac{p_1+ p_2 + p_4}{3}\\[0.5ex] 
\frac{p_0 + p_3 + p_5}{3} \\[0.5ex]
\frac{p_1 + p_2 + p_4}{3} \\[0.5ex]
\frac{p_0 + p_3 + p_5}{3}
\end{pmatrix}  \:\: , \:\: \omega_{9} =   \begin{pmatrix}
\frac{p_0 + p_4 + p_5}{3}\\[0.5ex] 
\frac{p_1 + p_2 + p_3}{3}\\[0.5ex] 
\frac{p_1+ p_2 + p_3}{3}\\[0.5ex] 
\frac{p_1 + p_2 + p_3}{3} \\[0.5ex]
\frac{p_0 + p_4 + p_5}{3} \\[0.5ex]
\frac{p_0 + p_4 + p_5}{3}  
\end{pmatrix}.
\end{split}
\end{align}
Consider now the $i$th population of the corresponding symmetrised estimate state 
\begin{align}
    \bra{i} \widetilde{\omega} \ket{i} &= \bra{i} \frac{1}{10}\sum^{9}_{n = 0} \omega_n \ket{i} = \frac{1}{30}(10p_i + 4(1 - p_i)) \\
    &=\frac{1}{30}(6p_i + 4).\nonumber
\end{align}
Observe that $p_i$ appears in each estimate in the $i$th position whilst each other term appears $4$ times in all the estimates in this position, giving the symmetrised estimate state $\widetilde{\omega} = \frac{1}{5}\left(\rho_\beta + 4 \frac{\mathbb 1}{6}\right)$. We can generalise this reasoning to arbitrary $d$ by thinking about what the $i$th entry of the symmetrised estimate will look like. Consider that in general for this setting we have $m = \frac{1}{2} \binom{d}{\frac{d}{2}, \frac{d}{2}} = \frac{d!}{(d/2)!(d/2)!2}$ estimates and that each term in the estimate is normalised by $\frac{d}{2}$ giving
\begin{align}
\bra{i} \widetilde{\omega} \ket{i} &= \bra{i} \frac{1}{m}\sum^{m - 1}_{n = 0} \omega_n \ket{i} = 
        \frac{2}{m d}\sum_{1 \leq \alpha_1 \leq \dots \leq \alpha_{d/2 - 1} \leq d - 1} p_i + p_{\alpha_1} + \dots + p_{\alpha_{d/2 - 1}}, \label{eq:gen_d_2}
\end{align}
where each $p_{\alpha_k} \in \mathcal{P} \setminus \{p_i\}$ is a choice of population from the set of populations $\mathcal P$ excluding $p_i$ and the indices $\alpha_k$ are such that each population appears once in a given sum of $d/2$ populations and permutations of sums are not over counted. Let us consider the combinatorics of these sums. Each sum can be represented by a set $\{a_0,\dots,a_{d/2 - 1}\}$ where each $a_k$ represents a contribution to the sum. Fixing the first term $a_0$ to be $p_i$, we wish to know how many sets there where the remaining $d/2 - 1$ slots are chosen from the remaining $d - 1$ populations. This is precisely the binomial coefficient $\binom{d - 1}{d/2 - 1}$ which we require to be equal to the number of estimates $m$ and indeed is after inspection. To see this consider
\begin{align}
    &\binom{d - 1}{d/2 - 1} = \frac{(d - 1)!}{(d/2)!(d/2 - 1)!} = \frac{(d - 1)!}{(d/2)(d/2 - 1)!(d/2 - 1)!} \\
\intertext{and}
    &m = \frac{1}{2}\binom{d}{\frac{d}{2},\frac{d}{2}} = \frac{d!}{2(d/2)!(d/2)!} = \frac{d(d-1)!}{2(d/2)(d/2-1)!(d/2)(d/2-1)!}
\end{align}
which are clearly equal expression. The number of times any other $p_k \neq p_i$ appears in the sums Eq.~\eqref{eq:gen_d_2} can also be found as the binomial expression $\binom{d-2}{d/2 -2}$ in this way Eq.~\eqref{eq:gen_d_2} becomes
\begin{align}
    \bra{i} \widetilde{\omega} \ket{i} &= \frac{2}{md}\left( \binom{d-1}{d/2 -1}p_i + \binom{d-2}{d/2 - 2}\left(\sum^{d-1}_{\stackrel{j =0}{i\neq j}}p_j\right) \right).
\end{align}
Recalling Pascal's identity $\binom{n}{k} = \binom{n - 1}{k -1} + \binom{n - 1}{k}$ and that $\sum^{d - 1}_{j = 0} p_j =1$, allows us to write Eq.~\eqref{eq:gen_d_2} as
\begin{align}
    \bra{i} \widetilde{\omega} \ket{i} &= \frac{2}{d\binom{d-1}{d/2 -1}}\left( \binom{d-2}{d/2 - 1}p_i + \binom{d-2}{d/2 - 2}\right),
\end{align}
and finally
\begin{equation}
    \widetilde{\omega} = \frac{2}{\binom{d-1}{d/2 -1}}\left( \frac{1}{d}\binom{d-2}{d/2 - 1}\rho_\beta + \binom{d-2}{d/2 - 2}\frac{\mathbb 1}{d}\right)
\end{equation}
a rather remarkably cute formula.
\subsection{$d = 6$, $k = 3$, $\vec{t} = \{1,2,3\}$}
\label{sec:d6k3}
Let us look at a special case, namely when $d=6$, $k=3$, and the measurement setting is $\vec{t} = (1,2,3)$. In this case, we know that the number of probes is $m = \binom{6}{3,2,1} = 60$. To derive a closed-form expression for the estimates, we start by noticing that the $m$ probes can be divide into $10$ groups, characterized by having a given selected level in its respective position, e.g., the first group is given by:
\begin{align}\label{Eq:all-probes-d-6}
\begin{split}
\omega_0 &=  \begin{pmatrix}
p_0\\[0.5ex] 
\frac{p_1 + p_2}{2}\\[0.5ex] 
\frac{p_1 + p_2}{2}\\[0.5ex] 
\frac{p_3+p_4+p_5}{3} \\[0.5ex]
\frac{p_3+p_4+p_5}{3} \\[0.5ex]
\frac{p_3+p_4+p_5}{3}  
\end{pmatrix}  \:\: , \:\: \omega_1 = \begin{pmatrix}
p_0\\[0.5ex] 
\frac{p_1 + p_3}{2}\\[0.5ex] 
\frac{p_2 + p_4 + p_5}{3}\\[0.5ex] 
\frac{p_1 + p_3}{2}\\[0.5ex] 
\frac{p_2 + p_4 + p_5}{3}\\[0.5ex] 
\frac{p_2 + p_4 + p_5}{3}  
\end{pmatrix}
 \:\: , \:\: \omega_2 =   \begin{pmatrix}
p_0\\[0.5ex] 
\frac{p_1 + p_4}{2}\\[0.5ex] 
\frac{p_2 + p_3 + p_5}{3}\\[0.5ex] 
\frac{p_2 + p_3 + p_5}{3}\\[0.5ex] 
\frac{p_1 + p_4}{2}\\[0.5ex] 
\frac{p_2 + p_3 + p_5}{3}  
\end{pmatrix}
  \:\: , \:\: \omega_3 =   \begin{pmatrix}
p_0\\[0.5ex] 
\frac{p_1 + p_5}{2}\\[0.5ex] 
\frac{p_2 + p_3 + p_4}{3}\\[0.5ex] 
\frac{p_2 + p_3 + p_4}{3}\\[0.5ex] 
\frac{p_2 + p_3 + p_4}{3}\\[0.5ex] 
\frac{p_1 + p_5}{2}  
\end{pmatrix}  \:\: , \:\: \omega_4 = \begin{pmatrix}
p_0\\[0.5ex] 
\frac{p_1 + p_3 + p_4}{3}\\[0.5ex] 
\frac{p_2 + p_5}{2}\\[0.5ex] 
\frac{p_1 + p_3 + p_4}{3}\\[0.5ex] 
\frac{p_1 + p_3 + p_4}{3}\\[0.5ex] 
\frac{p_2 + p_5}{2}  
\end{pmatrix} \:\: , \\[0.5ex]
\omega_5 &= \begin{pmatrix}
p_0\\ 
\frac{p_1 + p_2 + p_4}{3}\\[0.5ex] 
\frac{p_1 + p_2 + p_4}{3}\\[0.5ex] 
\frac{p_3 + p_5}{2}\\[0.5ex] 
\frac{p_1 + p_2 + p_4}{3}\\[0.5ex] 
\frac{p_3 + p_5}{2}  
\end{pmatrix}  \:\: , \:\: \omega_6 = \begin{pmatrix}
p_0\\[0.5ex] 
\frac{p_1 + p_2 + p_3}{3}\\[0.5ex] 
\frac{p_1 + p_2 + p_3}{3}\\[0.5ex] 
\frac{p_1 + p_2 + p_3}{3}\\[0.5ex] 
\frac{p_4 + p_5}{2}\\[0.5ex] 
\frac{p_4 + p_5}{2}  
\end{pmatrix}  \:\: , \:\: \omega_7 =   \begin{pmatrix}
p_0\\[0.5ex] 
\frac{p_1 + p_2 + p_5}{3}\\[0.5ex] 
\frac{p_1 + p_2 + p_5}{3}\\[0.5ex] 
\frac{p_3 + p_4}{2}\\[0.5ex] 
\frac{p_3 + p_4}{2}\\[0.5ex] 
\frac{p_1 + p_2 + p_5}{3} 
\end{pmatrix}  \:\: , \:\: \omega_8= \begin{pmatrix}
p_0\\[0.5ex] 
\frac{p_1 + p_4 + p_5}{3}\\[0.5ex] 
\frac{p_2 + p_3}{2}\\[0.5ex] 
\frac{p_2 + p_3}{2}\\[0.5ex] 
\frac{p_1 + p_4 + p_5}{3}\\[0.5ex] 
\frac{p_1 + p_4 + p_5}{3}  
\end{pmatrix}  \:\: , \:\: \omega_{9} =   \begin{pmatrix}
p_0\\[0.5ex] 
\frac{p_1 + p_3 + p_5}{3}\\[0.5ex] 
\frac{p_2 + p_4}{2}\\[0.5ex] 
\frac{p_1 + p_3 + p_5}{3}\\[0.5ex]
\frac{p_2 + p_4}{2}\\[0.5ex]
\frac{p_1 + p_3 + p_5}{3}  
\end{pmatrix}.
\end{split}
\end{align}
Now summing all probes given in Eq.~~\eqref{Eq:all-probes-d-6}, we can re-write their sum as
\begin{equation}\label{Eq:sum-probes-0}
    \sum_{i=0}^{9}\omega_i = \begin{pmatrix}
10 p_0\\ 
\frac{1}{2} \qty[5p_1+3(1-p_0)]\\ 
\frac{1}{2} \qty[5p_2+3(1-p_0)]\\ 
\frac{1}{2} \qty[5p_3+3(1-p_0)]\\ 
\frac{1}{2} \qty[5p_4+3(1-p_0)]\\ 
\frac{1}{2} \qty[5p_5+3(1-p_0)]
\end{pmatrix}.
\end{equation}
Notice that the remaining probes have the same structure as the ones given in Eq.~~\eqref{Eq:all-probes-d-6}, but with $p_0$ replaced by $p_i$, where $i\in{1,d-1}$, and its position corresponding to its $i$th label. Similarly, if we sum over them, we find the same structure as in Eq.~~\eqref{Eq:sum-probes-0}. Consequently, the symmetrised estimate in a setting $\v t = (1,2,3)$ is given by:
\begin{align}\label{Eq:estimate-123}
\begin{split}
\tilde{\omega}&=\frac{1}{60}\left[\begin{pmatrix}
10 p_0\\ 
\frac{1}{2} \left[3(1-p_0)+5 p_1\right]\\[0.5ex]  
\frac{1}{2} \left[3(1-p_0)+5 p_2\right]\\[0.5ex]  
\frac{1}{2} \left[3(1-p_0)+5 p_3\right]\\[0.5ex]  
\frac{1}{2} \left[3(1-p_0)+5 p_4\right]\\[0.5ex]  
\frac{1}{2} \left[3(1-p_0)+5 p_5\right]
\end{pmatrix}+\begin{pmatrix}
\frac{1}{2} \left[5 p_0+3(1-p_1)\right]\\[0.5ex]  
10 p_1\\ 
\frac{1}{2} \left[3(1-p_1)+5 p_2\right]\\[0.5ex]  
\frac{1}{2} \left[3(1-p_1)+5 p_3\right]\\[0.5ex]  
\frac{1}{2} \left[3(1-p_1)+5 p_4\right]\\[0.5ex]  
\frac{1}{2} \left[3(1-p_1)+5 p_5\right]
\end{pmatrix}+\begin{pmatrix}
\frac{1}{2} \left[5 p_0+3(1-p_2)\right]\\[0.5ex]  
\frac{1}{2} \left[5 p_1+3(1-p_2)\right]\\[0.5ex]  
10 p_2\\ 
\frac{1}{2} \left[3(1-p_2)+5 p_3\right]\\[0.5ex]  
\frac{1}{2} \left[3(1-p_2)+5 p_4\right]\\[0.5ex]  
\frac{1}{2} \left[3(1-p_2)+5 p_5\right]
\end{pmatrix}\right.\\&\hspace{0.8cm}+\left.\begin{pmatrix}
\frac{1}{2} \left[5 p_0+3(1-p_3)\right]\\[0.5ex]  
\frac{1}{2} \left[5 p_1+3(1-p_3)\right]\\[0.5ex]  
\frac{1}{2} \left[5 p_2+3(1-p_3)\right]\\[0.5ex]  
10 p_3\\ 
\frac{1}{2} \left[3(1-p_3)+5 p_4\right]\\[0.5ex]  
\frac{1}{2} \left[3(1-p_3)+5 p_5\right]
\end{pmatrix}+\begin{pmatrix}
\frac{1}{2} \left[5 p_0+3(1-p_4)\right]\\[0.5ex]  
\frac{1}{2} \left[5 p_1+3(1-p_4)\right]\\[0.5ex]  
\frac{1}{2} \left[5 p_2+3(1-p_4)\right]\\[0.5ex]  
\frac{1}{2} \left[5 p_3+3(1-p_4)\right]\\[0.5ex]  
10 p_4\\ 
\frac{1}{2} \left[3(1-p_4)+5 p_5\right]
\end{pmatrix}+\begin{pmatrix}
\frac{1}{2} \left[5 p_0+3(1-p_5)\right]\\[0.5ex]  
\frac{1}{2} \left[5 p_1+3(1-p_5)\right]\\[0.5ex]  
\frac{1}{2} \left[5 p_2+3(1-p_5)\right]\\[0.5ex]  
\frac{1}{2} \left[5 p_3+3(1-p_5)\right]\\[0.5ex]  
\frac{1}{2} \left[5 p_4+3(1-p_5)\right]\\[0.5ex]  
10 p_5
\end{pmatrix}\right].
\end{split}
\end{align}

Since the thermal state we want to learn is given by $\rho_{\beta}= (p_0, ..., p_5)^{\ms T}$, we can write Eq.~~\eqref{Eq:estimate-123} as
\begin{equation}\label{Eq:estimate-123-f}
    \tilde{\omega} = \frac{2}{5}\qty(\rho_{\beta}+\frac{1}{4}\iden_6).
\end{equation}

\section{Knowledge Concentration \label{App:knowlegde}}
\subsection{Warm-up - A Qutrit Example}
\label{sec:warm_up}
An agent attempting to estimate the state of a qutrit in a thermal state $\rho_\beta = (p_0,p_1,p_2)$ obtains a set of estimates $\Omega = \{\omega_0, \omega_1, \omega_2\}$ with values 
\begin{align}
    \Omega = \left\{\begin{pmatrix}
        p_0 \\ \frac{p_1 + p_2}{2} \\ \frac{p_1 + p_2}{2}
    \end{pmatrix},\begin{pmatrix}
        \frac{p_0 + p_2}{2} \\ p_1 \\ \frac{p_0 + p_2}{2}
    \end{pmatrix}, \begin{pmatrix}
        \frac{p_0 +p_1}{2} \\ \frac{p_0 + p_1}{2} \\ p_2
    \end{pmatrix} \right\} \label{eq:qutrit_estimate}
\end{align}
by sequentially correlating three qubit probes initially in the $\ket{0}_P$ state with the system and then generating the estimates in three qutrit memory states $\ket{0}_M$, following the protocol outlined in Sec~\ref{sec:est_gen_protocol}. The total knowledge the agent has obtained is captured by the product of the estimate states they have be able to obtain $\overline{\omega} = \omega_0 \otimes \omega_1 \otimes \omega_2.$ We wish to concentrate the knowledge obtained in each estimate into a single system and investigate when it is possible to reconstruct the state of $\rho_\beta$ in the first marginal $\tilde{\omega}_0 = \tr_{1,2}\{V \overline{\omega} V^\dagger\}$ of the transformed product of the estimate states. We have shown in Sec.~\ref{sec:know_conc} that such a unitary exists when $\overline{\omega}^\downarrow_0 \succ \rho_\beta$ and will make use of this Appendix to give an explicit example which illustrates this. 

To begin with, the unitaries $V$ we are interested in should create correlations in the global state $\overline{\omega}$ which do not create coherences in the marginals since the desired state $\rho_\beta$ is diagonal. One way to do this is to consider that the state space of the three qutrit estimates $\mathcal{H}_E = (\mathbb{C}^{\times 3})^{\otimes 3}$ may be partitioned into 9 subspaces 
\begin{align}
\label{eq:qutrit_decomp}
\mathcal{H}_{00} = \text{span}\{\ket{000}, \ket{111}, \ket{222} \} && \mathcal{H}_{01} = \text{span}\{\ket{001}, \ket{112}, \ket{220} \} && \mathcal{H}_{02} = \text{span}\{\ket{002}, \ket{110}, \ket{221} \}    \nonumber \\
\mathcal{H}_{10} = \text{span}\{\ket{010}, \ket{121}, \ket{202} \} && \mathcal{H}_{11} = \text{span}\{\ket{011}, \ket{122}, \ket{200} \} && \mathcal{H}_{12} = \text{span}\{\ket{012}, \ket{120}, \ket{201} \}    \\
\mathcal{H}_{20} = \text{span}\{\ket{020}, \ket{101}, \ket{212} \} && \mathcal{H}_{21} = \text{span}\{\ket{021}, \ket{102}, \ket{210} \} && \mathcal{H}_{22} = \text{span}\{\ket{022}, \ket{100}, \ket{211} \}     \nonumber
\end{align}
generally of the form
\begin{align}
    \mathcal{H}_{jk} = \text{span}\left\{\ket{i \, (i+j)\text{mod}\,3 \,\, (i+k)\text{mod}\, 3}\right\}^2_{i = 0}
\end{align}
giving $\mathcal{H}_E = \bigoplus^2_{j,k = 0} \mathcal{H}_{j,k}$. We can now observe that choosing $V = \bigoplus^2_{j,k = 0} V_{j,k}$ where each $V_{j,k}$ acts only on the space $\mathcal{H}_{j,k}$ will ensure that no global correlations result in off-diagonal terms in $\widetilde{\omega}_0$ when applying such a $V$ to $\overline{\omega}$. To see this, consider an arbitrary state $\eta$ satisfying this direct sum structure
\begin{align}
    \eta &= \bigoplus^2_{j,k = 0} \eta_{j,k} \\
    &= \bigoplus^2_{j,k = 0} \sum^2_{l,m = 0} q^{j,k}_{l,m} \ketbra{l \, (l+j)\text{mod}\,3 \, (l+k)\text{mod}\, 3}{m \, (m+j)\text{mod}\,3 \, (m+k)\text{mod}\, 3} 
\end{align}
where $q^{j,k}_{l,m}$ correspond to the $l,m$th coefficients of the $j,k$th block, and take a partial trace to find its first marginal.
\begin{align}
    \eta_0 &= \tr_{1,2} \{\eta\} = \sum^{2}_{r,s = 0} \bra{r \, s} \bigoplus^2_{j,k = 0} \eta_{j,k} \ket{r \, s} \\
           &= \sum^2_{j,k = 0} \left( \sum^2_{r,s,l,m = 0} q^{j,k}_{l,m} \ketbra{l}{m}  \braket{r \, s|(l+j)\text{mod}\,3 \, (l+k)\text{mod}\, 3}\braket{(m+j)\text{mod}\,3 \, (m+k)\text{mod}\, 3|r \, s} \right) \\
           &= \sum^2_{j,k = 0} \left( \sum^2_{r,s,l,m = 0} q^{j,k}_{l,m} \ketbra{l}{m}  \delta_{r s, (l+j)\text{mod}\,3 (l+k)\text{mod}\, 3} \, \delta_{(m+j)\text{mod}\,3 (m+k)\text{mod}\, 3 , r s} \right) \\
           &= \sum^2_{j,k = 0} \left( \sum^2_{l,m = 0} q^{j,k}_{l,m} \ketbra{l}{m}  \delta_{(l+j)\text{mod}\,3 (l+k)\text{mod}\, 3, (m+j)\text{mod}\,3 (m+k)\text{mod}\, 3} \right) 
\end{align}
This term has non-zero contributions when $r s = (l+j)\text{mod}\,3 (l+k)\text{mod}\, 3$ and $rs = (m+j)\text{mod}\,3 (m+k)\text{mod}\, 3$ which is only possible if $l = m$ giving
\begin{align}
\eta_0 = \sum^2_{j,k,l = 0}  q^{j,k}_{l} \ketbra{l}{l} \label{eq:marginal}
\end{align}
which is clearly diagonal state, with a sum of 9 contributions (one from each block) for each entry. In this way, any off-diagonal terms created in each block by $V_{j,k}$ will not contribute to off-diagonal terms in the first marginal, as we  desire. 

Let us now examine the application of this block-diagonal unitary $V$ on $\overline{\omega}$ to see if reconstruction is possible. Consider 
\begin{align}
    \widetilde{\omega} = \bigoplus^{2}_{j,k = 0} U_{j,k} \overline{\omega}_{j,k} U^{\dagger}_{j,k}
\end{align}
where 
\begin{align}\overline{\omega}_{j,k} = \sum^{2}_{i = 0} \bra{i \, (i+j)\text{mod}\,3 \, (i+k)\text{mod}\, 3} &\overline{\omega} \ket{i \, (i+j)\text{mod}\,3 \, (i+k)\text{mod}\, 3}\\
&\times\ketbra{i \, (i+j)\text{mod}\,3 \, (i+k)\text{mod}\, 3}\nonumber\end{align}
are the entries of $\overline{\omega}$ grouped according to the decomposition of the Hilbert Space as given in Eq.~\eqref{eq:qutrit_decomp}, which are diagonal since $\overline{\omega}$ is diagonal. Since these terms are diagonal it will prove useful to introduce a vector notation for them $\vec{r}_{j,k}$ with elements $r^{j,k}_{l}$ e.g. 
\begin{align}
    \vec{r}_{00} = \begin{pmatrix}
        \bra{000}\overline{\omega}\ket{000} \\ \bra{111}\overline{\omega}\ket{111} \\ \bra{222}\overline{\omega}\ket{222}
    \end{pmatrix}
\end{align}
and to note that the first marginal of $\overline{\omega}$ is the sum of these 9 vectors $\overline{\omega}_{0} = \tr_{1,2}\{\overline{\omega}\} = \sum^2_{j,k = 0} \vec{r}_{j,k}$, as evidenced by the form of Eq.~\eqref{eq:marginal}. These block-diagonal unitaries induce the impact of a unistochastic matrix on the vectorised form of $\overline{\omega}$ as from
\begin{gather}
    \widetilde{\omega} = \bigoplus^2_{j,k = 0} \sum^2_{l,m = 0} u^{j,k}_{l,m} {u^{j,k}_{m,l}}^* \ketbra{l}{m} \overline{\omega}_{j,k} \ketbra{m}{l},
\end{gather}
we identify the unistochastic matrix $M_{j,k} = u^{j,k}_{l,m} {u^{j,k}_{m,l}}^* \ketbra{l}{m} = |u^{j,k}_{l,m}|^{2}\ketbra{l}{m}$ and so since the first marginal is diagonal under the action of this unitary we may write
\begin{gather}
     \vec{\widetilde{\omega}}_0 = \sum^2_{j,k = 0} M_{j,k} \vec{r}_{j,k}.
\end{gather}
Restricting to the case where we set these 9 unistochastic matrices to be equal $M_{j,k} = \mathbf{M}$ we recover the majorisation condition presented in section \ref{sec:know_conc},
\begin{align}
\vec{\widetilde{\omega}}_0 = \mathbf{M} \sum^2_{j,k = 0} \vec{r}_{j,k} = \mathbf{M} \vec{\omega_0} \stackrel{?}{=} \rho_\beta
\end{align}
where such a unistochastic matrix $\mathbf{M}$ transforming the state to $\rho_\beta$ exists if $\overline{\omega}^\downarrow_0 \succ \rho_\beta$. For an explicit example, let the system be in the state $\rho_\beta = \frac{1}{6}(3,2,1)$ which is unknown to the agent. The agent would then obtain the set of estimates 
\begin{align}
    \Omega = \left\{\begin{pmatrix}
        \frac{1}{2} \\ \frac{1}{4} \\ \frac{1}{4}
    \end{pmatrix},\begin{pmatrix}
        \frac{1}{3} \\ \frac{1}{3} \\ \frac{1}{3}
    \end{pmatrix}, \begin{pmatrix}
        \frac{5}{12} \\ \frac{5}{12} \\ \frac{1}{6}
    \end{pmatrix} \right\}
\end{align}
and so the ordered marginal $\overline{\omega}^\downarrow_0 = (\frac{25}{48},\frac{5}{16},\frac{1}{6}) \succ \frac{1}{6}(3,2,1)$ which clearly majorises $\rho_\beta$ meaning a transformation $\mathbf{M}$ exists! In fact, using the \texttt{Mathematica} code~\cite{code} one finds that a suitable $\mathbf{M}$ would be 
\begin{gather}
    \mathbf{M} = \left(
\begin{array}{ccc}
 \frac{9}{10} & \frac{1}{10} & 0 \\
 \frac{1}{10} & \frac{9}{10} & 0 \\
 0 & 0 & 1 \\
\end{array}
\right).
\end{gather}
Using this code~\cite{code} one can check whether a given qutrit thermal state can be reconstructed by an agent using the knowledge concentration protocol described in Section~\ref{sec:know_conc} and what a suitable choice of unistochastic matrix to carry out this transformation could be.

Having checked the sufficient implication $(\Leftarrow)$ let us investigate the necessary implication $(\Rightarrow)$ i.e., if the existence of a unitary $V$ such that $\rho_\beta = \tr_{12}\{V\overline{\omega}V^\dagger\}$ implies 
$\overline{\omega}^\downarrow_0 \succ \rho_\beta$.

Let us call $V\overline{\omega}V^\dagger = \mu$. Relating the action of $V$ to a unistochastic matrix, Schur-Horn theorem~\cite{marshall2010inequalities} implies a majorisation condition $\text{diag}(\overline{\omega}) \succ \text{diag}(\mu)$
between the diagonal of $\overline{\omega}$ and $\overline{\mu}$ in a given basis which w.l.o.g we can choose to be the basis $\overline{\omega}$.

Note that that $\text{diag}(\overline{\omega})$ are the eigenvalues of $\overline \omega$ which let's in the decomposition $\mathcal{H} = \bigoplus^2_{j,k = 0} \mathcal{H}_{j,k}$ we expressed in the notation $\vec r$ showed that the first marginal of $\overline{\omega}$ can be expressed $\overline{\omega}_0 = \sum^2_{j,k = 0} \vec{r}_{j,k}$. We can also express $\rho_\beta$ in this way by writing out the partial trace of $\mu$ as
\begin{align}
    \mu_{j,k} = \sum^{2}_{i = 0} \bra{i \, (i+j)\text{mod}\,3 \, (i+k)\text{mod}\, 3} &\mu \ket{i \, (i+j)\text{mod}\,3 \, (i+k)\text{mod}\, 3} \label{eq:vec_decomp}\\
&\times\ketbra{i \, (i+j)\text{mod}\,3 \, (i+k)\text{mod}\, 3},\nonumber\end{align}
where $\sum^2_{j,k = 0} \mu_{j,k} = \rho_\beta$. That is for each population of $\rho_\beta$ we have $p_i = \sum^2_{j,k=0} \mu^{(i)}_{j,k}$ is the $i$th entry of each vector of form Eq.~\eqref{eq:vec_decomp} and similarly ${\overline{\omega}_0}_i =  \sum^2_{j,k=0} r^{(i)}_{j,k}$. But now from the majorisation condition recall that ordering every entry $r_{j,k,l}$  
\begin{align}
    r_m = r_{j,k,l} =  \bra{l \, (l+j)\text{mod}\,3 \, (l+k)\text{mod}\, 3} &\mu \ket{l \, (l+j)\text{mod}\,3 \, (l+k)\text{mod}\, 3}
\end{align}
in descending order denoted $r^\downarrow_{m}$ we have $\forall \, s \in [1,27],$
\begin{align}
    \sum^{s}_{m,n = 1} r^\downarrow_m \geq \mu^\downarrow_n.
\end{align}
Consider that since $\rho_\beta$ is a thermal state, its populations are in descending order and so $p_0 = \sum^9_{n=1} \mu^\downarrow_n$ whilst the 2nd and 3rd 9 entries form $p_1$ and $p_2$ respectively. Indeed we have for $j = {1,2,3}$
\begin{align}
    \sum^{9j}_{m,n = 1} r^\downarrow_m \geq \mu^\downarrow_n = p_j
\end{align}
which implies $\overline{\omega}^\downarrow \succ \rho_\beta$ as desired.
\subsection{Proof of the existence of Knowledge Concetrating Unitaries (Theorem IV.1)}
\label{sec:know_conc_proof}
For clarity of exposition we will here present the proof of the sufficient condition ($\Leftarrow$) before the proof of the necessary condition ($\Rightarrow$).

Beginning with the sufficient condition ($\Leftarrow$), here we discuss a framework for marginal transformation via unitary operations on $N-$qudit systems. In this approach, we decompose the diagonal elements of the single-qudit marginals into elements from different subspaces,  ensuring that any unitary transformation in these subspaces will not generate local off-diagonal elements in the marginals. More precisely, we consider an $N-$qudit system with an arbitrary diagonal initial state in the energy basis, i.e.,,
\begin{align}
    \overline{\omega}=\sum_{i_1,i_2,\hdots,i_N=0}^{d-1} \hspace{-2mm} p_{i_1i_2\hdots i_N}\ketbra{i_1,i_2,\hdots,i_N}{i_1,i_2,\hdots,i_N}.
\end{align}
Throughout the manuscripts, we only use the set of unitaries that do not generate any coherence in the marginals, it is convenient to represent the diagonal elements of the marginals with a vectorized notation. In the new notation, arbitrary diagonal state of the entire system can be represented as $\textbf{p}_{\text{tot}}= \, \text{diag}\{p_{i_1i_2\hdots i_N}\}_{i_1,i_2,\hdots,i_N=0}^{d-1}$ and the reduced state of $n-$th qudit $\overline{\omega}_{A_n}=\, \tr_{\bar{A}_n}{[ \overline{\omega}]}$, where ${\bar{A}_n}$ denotes to all the subsystems except $A_n$, are diagonal and can be represented via vector $\textbf{p}_{A_n}$. We then chose a particular vector decomposition to write $\textbf{p}_{A_n}$ as sum up $d^{N-1}$ $d-$dimensional vector according to 
\begin{align}
    \textbf{p}_{A_1} &=  \sum_{i_2,\hdots,i_N=0}^{d-1} \textbf{r}_{i_2,\hdots,i_N}\\
    \textbf{p}_{A_n} &= \sum_{i_2,\hdots,i_N=0}^{d-1} \Pi^{i_n}\,\textbf{r}_{i_2,\hdots,i_N},
\end{align}
where $\textbf{r}_{i_2,\hdots,i_N}= \sum_{j=0}^{d-1} p_{j,j+i_2,\hdots, j+ i_N} \textbf{e}_j $ and all the indices are to be taken in modulo $d$. In addition, $\{\textbf{e}_j\}$ is a set of orthonormal basis of $\mathbb{R}^d$ and $\Pi $ is the cyclic permutation matrix with elements $\Pi_{ij}= \delta_{i+1\, \text{mod} \, d,j}$. Furthermore, The decomposition of $\textbf{p}_{A_n}$s into these vectors corresponds to the selection of a total of $d^{N-1}$ subspaces,  $\mathcal{H}_{i_2,\hdots,i_N}$s, within the joint Hilbert space $\mathcal{H}_{\text{tot}}$,
\begin{equation}
    \mathcal{H}_{i_2,\hdots,i_N}=\, \text{span}\{\ket{j,j+i_2,\hdots, j+i_N}\}_{j=0}^{d-1}
\end{equation}
with $\mathcal{H}_{\text{tot}}=\oplus_{i_2,\hdots,i_N=0}^{d-1} \mathcal{H}_{i_2,\hdots,i_N}$, such that any unitary acts on these subspaces does not create coherence in all marginals if initial state of the entire system is diagonal or block diagonal with respect to the subspace decomposition, i.e.,, 
\begin{equation}
\overline{\omega}= \oplus_{\alpha} \overline{\omega}_\alpha,
\label{block diagonal form}
\end{equation}
where $\alpha:= (i_2,i_3,\hdots,i_N)$. This comes from the fact that 
due to the structure of the final state and the subspace spanned by the tensor product of orthogonal local bases—where each local basis contributes uniquely to the joint basis in each subspace—there are no off-diagonal elements in any of the marginals. To show this, let's only focus on $\overline{\omega}_\alpha $ in $\alpha-$th subspace. The general form of non-normalized $\overline{\omega}_\alpha $ can be written as 

\begin{equation}
    \overline{\omega}_{i_2,\hdots,i_N}= \sum_{j,k=0}^{d-1} p_{\tiny{(j,j+i_2,\hdots, j+ i_N),(k,k+i_2,\hdots, k+ i_N)}} \ketbra{j,j+i_2,\hdots, j+ i_N}{k,k+i_2,\hdots, k+ i_N}.
\end{equation}
 So it is straightforward to show that taking partial trace with respect to any of the subsystems generates delta function $\delta_{j,k}$. It means that all the reduced states are diagonal for such a state. Since this is true for all the subspaces, the marginals of any state in the form of \ref{block diagonal form} will be diagonal.   
It is worth mentioning that there are many subspace decompositions which have this property.  

Let us now consider a unitary transformation $\mu= U_{\text{tot}}\, \overline{\omega}\, U^\dagger_{\text{tot
}} $ with $U_{\text{tot}}= \oplus_{\alpha} U_{\alpha}$ where $\alpha:= (i_2,i_3,\hdots,i_N)$. If we start initially with a diagonal state $\overline{\omega}_{\text{tot}}= \oplus_{\alpha} \overline{\omega}_\alpha$, the final state can be written in the form of 
\begin{equation}
  \mu=  \oplus_{\alpha} U_{\alpha} \overline{\omega}_\alpha  U^\dagger_\alpha
\end{equation}
Due to the structure of the final state and the subspace spanned by the tensor product of orthogonal local bases—where each local basis contributes uniquely to the joint basis in each subspace—there are no off-diagonal elements in any of the marginals. In this case, it is convenient to use the vectorized notation of the diagonal elements to study the single-qudit marginal transformations. For any diagonal $\overline{\omega}_\alpha$, the marginal transformation  can be described by 
\begin{align}
 \textbf{p}_{A_1}  \to \bar{\textbf{p}}_{A_1} &=  \sum_{i_2,\hdots,i_N=0}^{d-1} M_{i_2,\hdots,i_N}\,\textbf{r}_{i_2,\hdots,i_N}\\
 \textbf{p}_{A_n}  \to \bar{\textbf{p}}_{A_n} &= \sum_{i_2,\hdots,i_N=0}^{d-1} \Pi^{i_n}\,{M}_{i_2,\hdots,i_N}\,\textbf{r}_{i_2,\hdots,i_N},
\end{align}
where $M_\alpha$s are $d\times d $ unistochastic matrices whose components are given by the elements of unitary matrices, i.e., $(M_\alpha)_{ij}= \vert (U_\alpha)_{ij}\vert^2$. 

Lets now focus on the state transformation of the first qudit. If we apply the same unitary at each subspace, i.e., $U_\alpha =U\, \forall \alpha$, we have
\begin{align}
 \bar{\textbf{p}}_{A_1} &=  \sum_{i_2,\hdots,i_N=0}^{d-1} M\,\textbf{r}_{i_2,\hdots,i_N}= M\,\sum_{i_2,\hdots,i_N=0}^{d-1} \textbf{r}_{i_2,\hdots,i_N}\\
 &=M \textbf{p}_{A_1},
\end{align}
By applying the Schur-Horn theorem and Hardy-Littlewood-Polya (HLP) theorem~\cite{marshall2010inequalities} which state that for any two vectors $\textbf{v}, \textbf{w} \in \mathbb{R}^d$, the majorisation condition $\textbf{v}\prec \textbf{w}$ iff there exists a doubly stochastic matrix $M$ such that $\textbf{v}=M \textbf{w}$, we can conclude that all states satisfying the majorisation condition can be reached. Specifically, this implies that  $\bar{\textbf{p}}_{A_1}\prec \textbf{p}_{A_1}$.

For the necessary condition ($\Rightarrow$) let $\bar{\omega} = \omega_0 \otimes \dots \otimes \omega_{n-1}$, where $\omega_i$ is a $d$-dimensional quantum state, $i \in \{0, \dots, n-1\}$.
Let $V$ be a unitary and $\mu = V \bar{\omega} V^{\dagger}$. Let 

\begin{equation}
 \rho_{\beta} = \Tr_{1,\dots,n-1} \left( V \bar{\omega} V^{\dagger} \right)= \Tr_{1,\dots, n-1}(\mu).
\end{equation}

As each $\omega_i$ is diagonal in $\ket{0}_i, \dots, \ket{d-1}_i$ we have
\begin{align}
\bar{\omega} &= \sum_{k_0, \dots, k_{n-1}=0}^{d-1} \bra{k_0, \dots, k_{n-1}} \bar{\omega} \ket{k_0, \dots, k_{n-1}}  \ket{k_0, \dots, k_{n-1}} \bra{{k_0, \dots, k_{n-1}}}\\
&=\sum_{k_0, \dots, k_{n-1}=0}^{d-1} \lambda_{k_0, \dots, k_{n-1}}^{\bar{\omega}}  \ket{k_0, \dots, k_{n-1}} \bra{{k_0, \dots, k_{n-1}}}
 \end{align}
 
 where
 
 \begin{equation}
     \lambda_{k_0, \dots, k_{n-1}}^{\bar{\omega}} := \bra{k_0, \dots, k_{n-1}} \bar{\omega} \ket{k_0, \dots, k_{n-1}}
 \end{equation}
 
 are the eigenvalues of $\bar{\omega}$. We denote by $\lambda_{k_0, \dots, k_{n-1}}^{\bar{\omega}, \downarrow}$ the $\lambda_{k_0, \dots, k_{n-1}}^{\bar{\omega}}$ ordered from biggest to smallest (according to the dictionary order). That is 
 \begin{equation}
     \lambda_{0,\dots,0,0}^{\bar{\omega}, \downarrow} \geq \lambda_{0,\dots,0,1}^{\bar{\omega}, \downarrow} \geq \dots \geq \lambda_{0,\dots,0,d-1}^{\bar{\omega}, \downarrow} \geq \lambda_{0,\dots,1,0}^{\bar{\omega}, \downarrow} \geq \dots \geq \lambda_{d-1,\dots,d-1,d-1}^{\bar{\omega}, \downarrow}
 \end{equation}
 and
 \begin{equation}
     \{\lambda_{k_0, \dots, k_{n-1}}^{\bar{\omega},\downarrow} \mid k_0,\dots, k_{n-1} \in \{0, \dots, d-1\}\} = \{\lambda_{k_0, \dots, k_{n-1}}^{\bar{\omega}} \mid k_0,\dots, k_{n-1} \in \{0, \dots, d-1\}\}. 
 \end{equation}
 Let 
 \begin{equation}
     \bar{\omega}^{\downarrow}:= \sum_{k_0,\dots,k_{n-1}}^{d-1} \lambda_{k_0, \dots, k_{n-1}}^{\bar{\omega}, \downarrow} \ket{k_0, \dots, k_{n-1}} \bra{k_0, \dots, k_{n-1}}
 \end{equation}
 and
 \begin{equation}
     \bar{\omega}^{\downarrow}_0=\Tr_{1 \dots k-1} (\bar{\omega}^{\downarrow}).
 \end{equation}
 Given an $n$-dimensional density matrix $\sigma=(\sigma_{ij})_{i,j=0}^{n-1}$ we will denote by $\diag(\sigma)=(\sigma_{00}, \sigma_{11},\dots,\sigma_{n-1.n-1})$ its vector of diagonal elements (ordered in the canonical order) and by $\diag(\sigma)$ the diagonal matrix with $\diag(\sigma)$ as diagonal. That is $\diag(\sigma)= (\sigma_{ij} \delta_{ij})_{i,j=0}^{n-1}$. We will furthermore denote by $\lambda^{\sigma}$ the vector of eigenvalues of $\sigma$ ordered from biggest to smallest.
 
 We want to prove the following
 \begin{lemma} \label{lemma:majorised}
     $\rho_{\beta} \prec \bar{\omega}_0^{\downarrow}$ (for any unitary $V$).
 \end{lemma}
 \begin{proof}
    We work in the computational basis $(\ket{k_0, \dots, k_{n-1}})$ and denote, in slight abuse of language, the density matrices in this basis by the density operator symbols. We do the proof in two steps. In step 1 we prove the assertion to be valid in the special case, where $\rho_{\beta}$ is a diagonal matrix with diagonal entries ordered in decreasing order. In step 2 we then show the general case by showing we can essentially reduce it to the above special case.

    Step 1. We here assume, as asserted above, that $\rho_{\beta}$ is diagonal and has diagonal entries ordered in decreasing order. Schur's theorem~\cite{marshall2010inequalities} tells us that the vector of eigenvalues of a given hermitian matrix majorises it's vector of diagonal elements. As $\mu= V \bar{\omega} V^{\dagger}$ is hermitian, applying that theorem to it we get
    \begin{equation}
        \diag(\mu) \prec \lambda^{\mu}.
    \end{equation}
    But as $\lambda^{\mu} = \diag(\bar{\omega}^{\downarrow})$ we get
    \begin{equation} \label{eq:SchurFormu}
        \diag{\mu} \prec \diag{\bar{\omega}^{\downarrow}}.
    \end{equation}
    Now, we need to check $\rho_{\beta} \prec \bar{\omega}^{\downarrow}_0$. That is we want to show
    \begin{equation}
        \lambda^{\rho_{\beta}} \prec \lambda^{\bar{\omega}_0^{\downarrow}}.
    \end{equation}
    Since both $\rho_{\beta}$ and $\bar{\omega}_0^{\downarrow}$ are diagonal matrices, what we need to show is in fact
    \begin{equation}
        \diag(\rho_{\beta}) \prec \diag(\bar{\omega}_0^{\downarrow}).
    \end{equation}
 Now since the states are normalised and ordered, that last majorisation relation holds iff
 \begin{equation}
     \sum_{k=0}^l \bra{k} \rho_{\beta} \ket{k} \leq \sum_{k=0}^l \bra{k} \bar{\omega}_0^{\downarrow} \ket{k}, \quad \forall l=0,\dots, d-2.
 \end{equation}
 But this is easy to show. Indeed, let $l \in \{0,\dots, d-2\}$. Then
 \begin{align}
     \sum_{k=0}^l \bra{k} \rho_{\beta} \ket{k} &= \sum_{k=0}^l \sum_{k_1, \dots, k_{n-1} =0}^{d-1} \bra{k k_1 \dots k_{n-1}} \mu \ket{k k_1 \dots k_{n-1}}\\
     &\leq \sum_{k=0}^l \sum_{k_1, \dots, k_{n-1} =0}^{d-1} \bra{k k_1 \dots k_{n-1}} \diag(\mu)^{\downarrow} \ket{k k_1 \dots k_{n-1}}\\
      &\leq \sum_{k=0}^l \sum_{k_1, \dots, k_{n-1} =0}^{d-1} \bra{k k_1 \dots k_{n-1}} \bar{\omega}^{\downarrow} \ket{k k_1 \dots k_{n-1}}\\
      &= \sum_{k=0}^l \bra{k} \bar{\omega}_0^{\downarrow} \ket{k},
 \end{align}
where $\diag(\mu)$ is the diagonal matrix that has the diagonal elements of $\mu$ in the diagonal. And $\diag(\mu)^{\downarrow}$ is $\diag(\mu)$ where the diagonal elements are reodered from biggest to smallest (in the canonical order). The inequality of the second line holds because the sum of $p$ elements of a list of values is always smaller than the sum of the greatest $p$ elements of said list of values. The inequality of the third line holds thanks to equation~\ref{eq:SchurFormu}. This concludes step 1.

We now turn our attention to step 2. For this suppose now that $\rho_{\beta}$ is arbitrary. Let $U$ be the unitary diagonalising $\rho_{\beta}$ and let
\begin{equation}
    \tilde{\rho}_{\beta} = U \rho_{\beta} U^{\dagger}.
\end{equation}
We also choose $U$ such that
\begin{equation}
    \bra{0} \tilde{\rho}_{\beta} \ket{0} \geq \bra{1} \tilde{\rho}_{\beta} \ket{1} \geq \dots \geq \bra{d-1} \tilde{\rho}_{\beta} \ket{d-1}.
\end{equation}
Now
\begin{align}
    \tilde{\rho}_{\beta} &= U \Tr_{1 \dots n-1} (\mu) U^{\dagger}\\
    &= \Tr_{1 \dots n-1} \left( \tilde{\mu} \right),
\end{align}
where $\tilde{\mu} = \tilde{V} \bar{\omega} \tilde{V}^{\dagger}$,
with $\tilde{V}= U \otimes \mathds{1}_{1, \dots, n-1} V$. Note that $\tilde{V}$ is a unitary and that $\tilde{\mu}$ is hermitian with
\begin{equation}
    \lambda^{\tilde{\mu}}= \diag{\bar{\omega}^{\downarrow}}.
\end{equation}
We are therefore in the situation of step 1 with $\tilde{\rho}_{\beta}$ (and $\tilde{\mu}$) in place of $\rho_{\beta}$ (and $\mu$). Using step 1 we hence get
\begin{equation}
    \tilde{\rho}_{\beta} \prec \bar{\omega}_0^{\downarrow}.
\end{equation}
That is 
\begin{equation}
    \lambda^{\tilde{\rho}_{\beta}} \prec \lambda^{\bar{\omega}_0^{\downarrow}}.
\end{equation}
And as $\lambda^{\tilde{\rho}_{\beta}}= \lambda^{\rho_{\beta}}$ we get
\begin{equation}
    \lambda^{\rho_{\beta}} \prec \lambda^{\bar{\omega}_0^{\downarrow}}.
\end{equation}
That is
\begin{equation}
    \rho_{\beta} \prec \bar{\omega}_0^{\downarrow},
\end{equation}
as desired.
\end{proof}

There is a small Corollary to this that we want to state now. To this let
\begin{align}
    A &= \{ \rho_{\beta} \mid \rho_{\beta}= \Tr_{1, \dots, n-1}(V \bar{\omega} V^{\dagger}), \; V \text{ unitary}\},\\
    B &= \{ \sigma \mid \sigma \prec \bar{\omega}^{\downarrow}_0 \}.
\end{align}
$A$ is the set of all states reachable on 0 upon applying a global unitary to $\bar{\omega}$. $B$ is the set of all states (on 0) that are majorised by $\bar{\omega}_0$.
\begin{corollary}
    $A \subset B$.
\end{corollary}
\begin{proof}
    Let $\rho_{\beta} \in A$. Then using Lemma~\ref{lemma:majorised}, $\rho_{\beta} \prec \bar{\omega}_0^{\downarrow}$. That is $\rho_{\beta} \in B$. As $\rho_{\beta}$ was an arbitrary element of $A$ we just proved $A \subset B$, as desired.
\end{proof}

\section{Work Extraction from unknown states \label{A:applications}}

\begin{figure*}[h]
        \centering
	\includegraphics{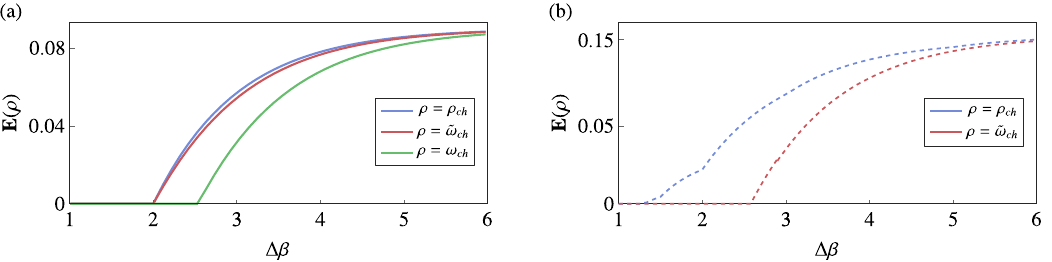}
	\caption{\justifying{ 
 \textbf{Work extraction from unknown states}. For two thermal states with inverse temperatures $\beta_c$ and $\beta_h$, and an equidistant energy spectrum, we plot the ergotropy as a function of the inverse temperature difference $\Delta \beta := \beta_c - \beta_h$. In panel (a), for a three-dimensional thermal state, the blue curve shows the ergotropy extracted with full knowledge of the state, while the red curve represents the ergotropy using the symmetrized estimate. The green curve shows the ergotropy without symmetrisation, yielding less extractable work. In panel (b), for a six-level thermal state, the gap between the actual ergotropy (dashed blue) and the estimate (dashed red) increases. As dimensionality grows, the symmetrized estimate becomes less accurate, as information is more spread across the system}}
\label{F-work-extraction}
\end{figure*}

We now study a paradigmatic task in both quantum and classical thermodynamics: the extraction of work from two coupled systems. Consider the scenario where an agent is given two thermal states at different, unknown temperatures, and its task is to extract the optimal amount of work from these two subsystems. The agent does not have any information about the population of the joint state; instead, it is allowed to use probes to reconstruct imperfect copies of the two states. How reliably can the agent extract work from the two subsystems by using its probes to obtain an optimal protocol for work extraction?

In order to address this question, we begin by considering two subsystems, each initially prepared in thermal (Gibbs) states, $\rho_{\beta_X} = \frac{e^{-\beta_X H}}{\tr(e^{-\beta_X H})}$, where $\beta_X$ denotes the inverse temperature of the subsystem $\ms{X} \in \{c, h\}$, representing the cold and hot subsystems, respectively. Each thermal state is described by a local Hamiltonian $H_{\ms X} = \sum_{i=1}^{d_{\ms X}} \epsilon^{\ms{X}}_i \ketbra{i}{i}_{\ms X}$ of dimension $d_{\ms X}$, where $\epsilon^{\ms X}_i$ indicates the $i$th energy eigenvalue and $\ket{i}_{\ms X}$ the corresponding eigenstate of the subsystem $\ms{X}$. As a result, the state of the composite system is given by 
\begin{equation}
   \rho_{ch}:=\rho_{\beta_c} \otimes \rho_{\beta_h} = \frac{e^{-\beta_c H_c}}{Z_c}\otimes  \frac{e^{-\beta_ h H_h}}{Z_h},
\end{equation}
where $Z_{\ms{X}} = \tr(e^{-\beta_X H_X})$ is the partition function for subsystem $\ms{X}$. 
Since the agent has no prior knowledge of $\beta_c$ and $\beta_h$, other than the fact that the two inverse temperatures differ, the optimal strategy is to distribute the information of $\rho_{ch}$ by symmetrising over two probes. This leads to the construction of the symmetrised state $\widetilde{\omega}_{\ms{X}} = \frac{1}{d-1}[\rho_{\beta_{\ms X}} + \frac{(d-2)}{d}\mathbbm{1}_d]$, where $\mathbbm{1}_d$ is the identity matrix of dimension $d$. As a result, the agent’s estimate of the composite state is $\widetilde{\omega}_{ch}:= \widetilde{\omega}_c \otimes \widetilde{\omega}_h$.

As a figure of merit, we adopt the notion of ergotropy~\cite{allahverdyan2004maximal}, which quantifies the maximum amount of work that can be extracted from a quantum system through cyclic and unitary operations. It is defined as the minimisation over all possible unitary operations  $U$, such that the energy exchange is minimized
\begin{equation}\label{Eq:ergotropy}
    \mathbf{E}(\rho):= \tr(\rho H) - \underset{U}{\min} \tr[H(U\rho U^{\dagger})].
\end{equation} 

Our protocol for work extraction follows the following steps:
\begin{enumerate}
    \item  First, we construct the optimal work-extraction unitary, $U^{\star}_{\tilde{\omega}_{ch}}$, which maximizes the ergotropy with respect to the symmetrized state $\widetilde{\omega}_{ch}$.
    \item Second, we apply $U^{\star}_{\tilde{\omega}_{ch}}$ to the unknown thermal state $\rho_{ch}$ to estimate the amount of work that can be extracted from the unknown state. Thus, we are interested in the following quantity:
    \begin{equation}
       \hspace{0.6cm} \mathbf{E}(\rho_{ch}):= E_{c}+E_{h} - \tr[H_{ch}(U^{\star}_{\tilde{\omega}_{ch}}\rho_{ch} U^{\star}{_{\tilde{\omega}_{ch}}}^{\dagger})].
    \end{equation}
    where $E_{\ms X}:= \tr(H_{\ms X} \rho_{\beta_{\ms X}})$ is the average energy of each subsystem and $H_{ch}:= H_c\otimes \iden + \iden\otimes H_h$ is the Hamiltonian of the composite system.
\end{enumerate}

Naturally, this protocol results in extracting an amount of work that is equal to or lower than the actual value, as the symmetrized state only provides an imperfect estimate of the true system, leading to suboptimal work extraction. 

Despite the fact that the agent only has access to an estimate of the true state, the work extracted using this symmetrized approximation is still quite close to the optimal value. As shown in Figure~\hyperref[F-work-extraction]{\ref{F-work-extraction}a} the ergotropy of the estimated state (red curve) closely follows the ergotropy of the actual thermal state (in blue curve). A possible explanation is that the optimal unitary requires the populations of the estimate and the actual state to be ordered similarly. When the populations are ordered in the same way, the optimal protocol will perform the same series of swaps, even if the population values themselves differ slightly. Conversely, if no symmetrisation is performed (see panel~\hyperref[F-work-extraction]{\ref{F-work-extraction}b}), the ergotropy (green curve) is significantly lower than the actual value, resulting in worse performance and highlighting the advantage of symmetrisation. In this case, the information is concentrated in specific probes rather than being spread across both. However, in higher-dimensional systems, as shown in panel~\hyperref[F-work-extraction]{\ref{F-work-extraction}c}, using lower-dimensional probes is not the best approach. Spreading the information across the estimate's entries is not always optimal, as reconstructing $\rho_{ch}$ becomes increasingly difficult as the dimensionality grows.

An agent's ability to extract work from an unknown source was treated in a recent work~\cite{safranek_2023} where the agent extracting work using coarse-grained unitaries was shown to have the same work extraction capabilites as a bling agent using Haar random unitaries. In contrast with this, our symmetrisiation protocol approaches the ergotropic value.

\section{Implications for the resource theory of athermality.}
\label{A:resource_theory}
\begin{figure*}[h]
       \centering
	\includegraphics{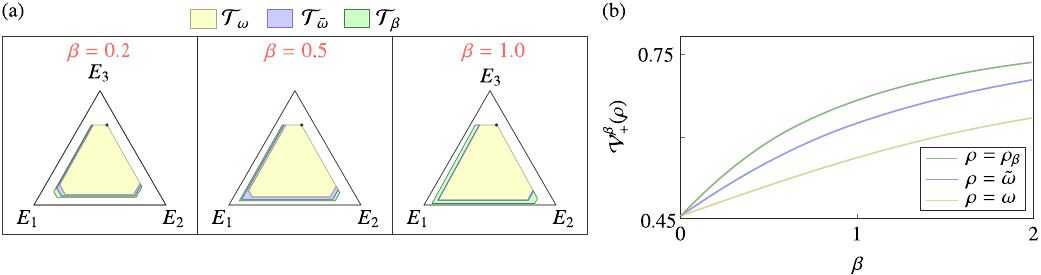}
	\caption{\justifying{ 
 \textbf{Imperfect state transformation}. For a three-level system with populations $\v{p} = (0.1, 0.2, 0.7)$, represented by the black dot $\bullet$, and energy spectrum $E_1 = 0$, $E_2 = 1$, and $E_3 = 2$, we plot its imperfect future cone of reachable states under thermal operations. In panel (b), we display the volume of the imperfect and perfect future thermal cones.}}
\label{F-state-transformation}
\end{figure*}
Our results also motivate a discussion of hidden cost assumptions in quantum resource theories~\cite{chitambarGour}. More precisely, the basic elements of the resource-theoretic framework are \emph{free} operations and \emph{free} states which are available to an agent at no cost. But this begs the question-- \textit{what are the resources behind perceiving these states and operations as free}? In our view the agent must at the very least have knowledge of these free states i.e. be able to reconstruct these free states unitarily. If they have less resources and are not able to perfectly reconstruct these states then their set of free states changes.  The resource theory of thermodynamics~\cite{Janzing2000,horodecki2013fundamental, brandao2015second,Lostaglio2019} gives a complete set of laws for characterising general state transformations under thermodynamic constraints. The downside of this formalism is that it typically assumes that $\rho_{\beta}$ is a free state which is left invariant by the free thermal operations. But as we have shown in this work there is a cost of to having knowledge of a thermal state. Our framework allows us to address the question of how state transformations are affected when there is no precise information about the ``free" thermal state, but only an approximate estimation of it. This can be translated into allowing the agent access to probes, which help reconstruct the unknown thermal state and compare the achievability of states with the actual set of achievable transformations. 

It is known that the set of achievable states via thermal operations from a given incoherent initial state can be fully characterized using the notion of thermomajorisation~\cite{horodecki2013fundamental}. The so-called future thermal cone, $\mathcal{T}_{\beta}(\rho)$~\cite{deoliveirajunior2022}. Consequently, to understand how thermodynamic transformations are affected by the inability to access a specific thermal state $\rho_{\beta}$, we could define the imperfect future thermal cone $\mathcal{T}_{\bar{\omega}}$ after symmetrisation and for a single estimate $\mathcal{T}_{\omega}$. Note that both imperfect future thermal cones are still characterized by the same set of rules as the standard cone. The only difference lies in the relative thermal state used to construct its extreme points. In Fig.~\ref{F-state-transformation}, we give a gist of how state transformations would be modified under imperfect knowledge of the thermal state. As observed, both the symmetrised estimate and the single estimate are very close to the actual future thermal cone, though the symmetrised estimate approaches relatively better than for a single estimate. Of course, this is a rough approximation, as we assume the system's thermal state is unknown and do not directly refer to the thermal bath. Similarly, in panel \hyperref[F-work-extraction]{\ref{F-state-transformation}b}, we compare the volumes of these regions to quantify how closely the volumes of the symmetrised and single estimates approach that of the actual thermal cone.

\section{Thermodynamic Considerations \label{Sec:thermodynamic-considerations}}

In this section, we briefly give insight into the thermodynamic cost associated with extracting information from a system by correlating it with a probe of limited dimension and generating an estimate based on these correlations.

To examine the thermodynamic cost associated with this process, we will use the notion of entropy production~\cite{landi_entropy}. This concept measures the irreversibility of a process and allows us to infer how much heat was irreversibly dissipated throughout the process. For a unitary process between a system $\ms S$ and a reservoir $\ms R$, the entropy production can be expressed as a combination of the mutual information between the system and the reservoir at the end of the process, and the relative entropy of the reservoir before and after the process~\cite{Esposito_2010,Reeb_2014}
\begin{gather}
    \Sigma = \mathcal{I}(\ms{S'}\colon \ms{R'}) + \mathcal{D}(\rho'_R ||\rho_R), \label{eq:ent_prod}
\end{gather}
where $\mathcal{I}(\ms{S'}\colon \ms{R'}) = S(\rho'_\ms{S}) + S(\rho'_{\ms R}) - S(\rho'_{\ms{SR}})$, $\mathcal{D}(\rho'_\ms{R} ||\rho_\ms{R}) = \tr(\rho'_\ms{R} \log \rho_\ms{R}) - S(\rho'_\ms{R})$ and $S(\rho) = -\tr\{\rho \log \rho\}$ is the von Neumann entropy.

Starting with the information extraction step $U_{IE}(\ketbra{0}_k \otimes \rho_\beta)U_{IE}$ let us take $\ms S$ to be the probe $\ketbra{0}_k$ and $\ms R$ to be the unknown system $\rho_\beta$, we observe that this correlating process leaves the system unchanged and any entropy production will result from a change in the entropy of the probe. We assume that the main system, which has a larger dimension than the probe, acts as the reservoir and so the relative entropy term vanishes. Similarly, the mutual information term reduces to the entropy of the probe after the information extraction step, so that the entropy production becomes simply $\Sigma = \Delta S_\ms{P} = S(\sigma)$. As a result, the more coarse-grained the probe is—i.e., the larger the difference between $k$ and $d$--and the more uniform the coupling $\vec t$, the larger the dissipation in the information extraction step. 

For the estimate generation step, we again treat the larger system—this time the memory—as the reservoir and the probe as the system in this process. In this case, the relative entropy of the reservoir's states before and after the process is
\begin{gather}
    \mathcal{D}(\omega || \ketbra{0}{0}_d) = -S(\omega)-\tr(\omega \log \ketbra{0}{0}_d),
\end{gather}
which is well-defined since $\tr(\omega \ketbra{0}{0}_d) > 0$, but diverges for the zero eigenvalues of $\ketbra{0}{0}_d$. This can be seen as related to the fact that preparing a pure state requires an infinite thermodynamic cost~\cite{Guryanova_2020,taranto_2020,Buffoni2023cooperativequantum}. Instead, let us consider a near-perfect memory $\rho_M = (1 - \epsilon)\ketbra{0}{0}_d + \frac{\epsilon}{d}\mathbb{1}_d$ where focusing on the trace gives
$
\tr(\omega \log \rho_M) = \frac{\overline p_0}{t_0}\log(1-\epsilon) + (t_0 - 1)\frac{\overline p_0}{t_0} \log(\frac{\epsilon}{d}) + \sum^{k - 1}_{i = 1} \frac{\overline p_i}{t_i} t_i \log(\frac{\epsilon}{d})$
and so the relative entropy
\begin{align} \mathcal{D}(\omega || \rho_M) &= -S(\omega) - \frac{\overline p_0}{t_0} \left(\log(1-\epsilon)\left(\frac{\epsilon}{d}\right)^{\frac{t_0 -1}{t_0}}\right) \nonumber  \\ & \hspace{1.55cm} - (1 - \overline{p_0})\log(\frac{\epsilon}{d}).
\end{align} 
Finally, the mutual information in this case takes the form of
\begin{align}
    \mathcal{I}(\ms{S'}:\ms{R'}) &= S(\sigma') + S(\omega) - S(\sigma) \\
    &= \Delta S_{\ms{P}'} + \Delta S_{M} \nonumber
\end{align}
where we use the facts that entropy is invariant under unitary processes and additive under tensor product. Consequently, the heat dissipated during the estimate generation can be written as
\begin{align}
\Sigma = \Delta S_{P'} &+ \frac{\overline p_0}{t_0} \left(\log\left[(1-\epsilon)\left(\frac{\epsilon}{d}\right)^{\frac{t_0 -1}{t_0}}\right]\right) \label{eq:diss_est_gen}\\
&- (1 - \overline{p_0})\log(\frac{\epsilon}{d}). \nonumber
\end{align}
Therefore, we observe that greater dissipation occurs when the original state of the unknown system is colder—i.e., more pure—resulting in a larger contribution from $\overline{p}_0$ contribution. More generally, in the dissipation of both the information extraction and estimate generation processes, we note that if the probe or memory were less pure, and thus more thermal, the temperature difference from the state they are converted to would be smaller, leading to lower dissipation, but the resulting estimate would be less accurate. It is good to note that while the examination of the thermodynamic cost involved in the acquisition of quantum knowledge is novel and far from addressed in these initial observations-- the examination of the thermodynamic costs inherent in the quantum acquisition of classical knowledge is a much studied field~\cite{brillouin1953negentropy,PhysRevLett.102.250602,Guryanova_2020,wilde_triple_trade_off,Debarba2024}.We expect that in a full account of the thermodynamics of quantum knowledge acquisition the main thermodynamic cost will be due to state preparation of near pure probes and memories and any dissipation throughout the process will be small in comparison.


\begin{thebibliography}{49}%
	\makeatletter
	\providecommand \@ifxundefined [1]{%
		\@ifx{#1\undefined}
	}%
	\providecommand \@ifnum [1]{%
		\ifnum #1\expandafter \@firstoftwo
		\else \expandafter \@secondoftwo
		\fi
	}%
	\providecommand \@ifx [1]{%
		\ifx #1\expandafter \@firstoftwo
		\else \expandafter \@secondoftwo
		\fi
	}%
	\providecommand \natexlab [1]{#1}%
	\providecommand \enquote  [1]{``#1''}%
	\providecommand \bibnamefont  [1]{#1}%
	\providecommand \bibfnamefont [1]{#1}%
	\providecommand \citenamefont [1]{#1}%
	\providecommand \href@noop [0]{\@secondoftwo}%
	\providecommand \href [0]{\begingroup \@sanitize@url \@href}%
	\providecommand \@href[1]{\@@startlink{#1}\@@href}%
	\providecommand \@@href[1]{\endgroup#1\@@endlink}%
	\providecommand \@sanitize@url [0]{\catcode `\\12\catcode `\$12\catcode `\&12\catcode `\#12\catcode `\^12\catcode `\_12\catcode `\%12\relax}%
	\providecommand \@@startlink[1]{}%
	\providecommand \@@endlink[0]{}%
	\providecommand \url  [0]{\begingroup\@sanitize@url \@url }%
	\providecommand \@url [1]{\endgroup\@href {#1}{\urlprefix }}%
	\providecommand \urlprefix  [0]{URL }%
	\providecommand \Eprint [0]{\href }%
	\providecommand \doibase [0]{https://doi.org/}%
	\providecommand \selectlanguage [0]{\@gobble}%
	\providecommand \bibinfo  [0]{\@secondoftwo}%
	\providecommand \bibfield  [0]{\@secondoftwo}%
	\providecommand \translation [1]{[#1]}%
	\providecommand \BibitemOpen [0]{}%
	\providecommand \bibitemStop [0]{}%
	\providecommand \bibitemNoStop [0]{.\EOS\space}%
	\providecommand \EOS [0]{\spacefactor3000\relax}%
	\providecommand \BibitemShut  [1]{\csname bibitem#1\endcsname}%
	\let\auto@bib@innerbib\@empty
	\bibitem [{\citenamefont {Fermi}(1956)}]{fermi}%
	\BibitemOpen
	\bibfield  {author} {\bibinfo {author} {\bibfnamefont {E.}~\bibnamefont {Fermi}},\ }\href {https://books.google.pl/books?id=VEZ1ljsT3IwC} {\emph {\bibinfo {title} {Thermodynamics}}}\ (\bibinfo  {publisher} {Dover Publications},\ \bibinfo {year} {1956})\BibitemShut {NoStop}%
	\bibitem [{\citenamefont {Callen}(1985)}]{callen}%
	\BibitemOpen
	\bibfield  {author} {\bibinfo {author} {\bibfnamefont {H.}~\bibnamefont {Callen}},\ }\href {https://books.google.at/books?id=XJ0RAQAAIAAJ} {\emph {\bibinfo {title} {Thermodynamics and an Introduction to Thermostatistics}}},\ Thermodynamics and an Introduction to Thermostatistics\ (\bibinfo  {publisher} {Wiley},\ \bibinfo {year} {1985})\BibitemShut {NoStop}%
	\bibitem [{\citenamefont {Knott}(1911)}]{maxwell1867}%
	\BibitemOpen
	\bibfield  {author} {\bibinfo {author} {\bibfnamefont {C.~G.}\ \bibnamefont {Knott}},\ }\href {https://doi.org/https://www.maths.ed.ac.uk/~v1ranick/papers/taitbio.pdf} {\emph {\bibinfo {title} {Life and scientific work of Peter Guthrie Tait}}},\ Vol.~\bibinfo {volume} {1}\ (\bibinfo  {publisher} {Cambridge university press},\ \bibinfo {year} {1911})\BibitemShut {NoStop}%
	\bibitem [{\citenamefont {Maruyama}\ \emph {et~al.}(2009)\citenamefont {Maruyama}, \citenamefont {Nori},\ and\ \citenamefont {Vedral}}]{vedral_review}%
	\BibitemOpen
	\bibfield  {author} {\bibinfo {author} {\bibfnamefont {K.}~\bibnamefont {Maruyama}}, \bibinfo {author} {\bibfnamefont {F.}~\bibnamefont {Nori}},\ and\ \bibinfo {author} {\bibfnamefont {V.}~\bibnamefont {Vedral}},\ }\bibfield  {title} {\bibinfo {title} {Colloquium: The physics of maxwell's demon and information},\ }\href {https://doi.org/10.1103/RevModPhys.81.1} {\bibfield  {journal} {\bibinfo  {journal} {Rev. Mod. Phys.}\ }\textbf {\bibinfo {volume} {81}},\ \bibinfo {pages} {1} (\bibinfo {year} {2009})}\BibitemShut {NoStop}%
	\bibitem [{\citenamefont {de~Oliveira~Junior}\ \emph {et~al.}(2025)\citenamefont {de~Oliveira~Junior}, \citenamefont {Brask},\ and\ \citenamefont {Chaves}}]{junior2025friendlyguideexorcisingmaxwells}%
	\BibitemOpen
	\bibfield  {author} {\bibinfo {author} {\bibfnamefont {A.}~\bibnamefont {de~Oliveira~Junior}}, \bibinfo {author} {\bibfnamefont {J.~B.}\ \bibnamefont {Brask}},\ and\ \bibinfo {author} {\bibfnamefont {R.}~\bibnamefont {Chaves}},\ }\bibfield  {title} {\bibinfo {title} {A friendly guide to exorcising maxwell's demon},\ }\href {https://doi.org/10.1103/phkv-wrsd} {\bibfield  {journal} {\bibinfo  {journal} {PRX Quantum}\ }\textbf {\bibinfo {volume} {6}},\ \bibinfo {pages} {030201} (\bibinfo {year} {2025})}\BibitemShut {NoStop}%
	\bibitem [{\citenamefont {Landauer}(1961)}]{landauer1961}%
	\BibitemOpen
	\bibfield  {author} {\bibinfo {author} {\bibfnamefont {R.}~\bibnamefont {Landauer}},\ }\bibfield  {title} {\bibinfo {title} {Irreversibility and heat generation in the computing process},\ }\href {https://worrydream.com/refs/Landauer_1961_-_Irreversibility_and_Heat_Generation_in_the_Computing_Process.pdf} {\bibfield  {journal} {\bibinfo  {journal} {IBM J. Res. Dev.}\ }\textbf {\bibinfo {volume} {5}},\ \bibinfo {pages} {183} (\bibinfo {year} {1961})}\BibitemShut {NoStop}%
	\bibitem [{\citenamefont {Bennett}(1982)}]{bennett1982}%
	\BibitemOpen
	\bibfield  {author} {\bibinfo {author} {\bibfnamefont {C.~H.}\ \bibnamefont {Bennett}},\ }\bibfield  {title} {\bibinfo {title} {The thermodynamics of computation—a review},\ }\href {https://sites.cc.gatech.edu/computing/nano/documents/Bennett%20-%20The%20Thermodynamics%20Of%20Computation.pdf} {\bibfield  {journal} {\bibinfo  {journal} {Int. J. Theor. Phys.}\ }\textbf {\bibinfo {volume} {21}},\ \bibinfo {pages} {905} (\bibinfo {year} {1982})}\BibitemShut {NoStop}%
	\bibitem [{\citenamefont {Goold}\ \emph {et~al.}(2016)\citenamefont {Goold}, \citenamefont {Huber}, \citenamefont {Riera}, \citenamefont {del Rio},\ and\ \citenamefont {Skrzypczyk}}]{Goold_2016}%
	\BibitemOpen
	\bibfield  {author} {\bibinfo {author} {\bibfnamefont {J.}~\bibnamefont {Goold}}, \bibinfo {author} {\bibfnamefont {M.}~\bibnamefont {Huber}}, \bibinfo {author} {\bibfnamefont {A.}~\bibnamefont {Riera}}, \bibinfo {author} {\bibfnamefont {L.}~\bibnamefont {del Rio}},\ and\ \bibinfo {author} {\bibfnamefont {P.}~\bibnamefont {Skrzypczyk}},\ }\bibfield  {title} {\bibinfo {title} {The role of quantum information in thermodynamics—a topical review},\ }\href {https://doi.org/10.1088/1751-8113/49/14/143001} {\bibfield  {journal} {\bibinfo  {journal} {J. Phys. A-Math.}\ }\textbf {\bibinfo {volume} {49}},\ \bibinfo {pages} {143001} (\bibinfo {year} {2016})}\BibitemShut {NoStop}%
	\bibitem [{\citenamefont {Binder}\ \emph {et~al.}(2019)\citenamefont {Binder}, \citenamefont {Correa}, \citenamefont {Gogolin}, \citenamefont {Anders},\ and\ \citenamefont {Adesso}}]{binder2019thermodynamics}%
	\BibitemOpen
	\bibfield  {author} {\bibinfo {author} {\bibfnamefont {F.}~\bibnamefont {Binder}}, \bibinfo {author} {\bibfnamefont {L.}~\bibnamefont {Correa}}, \bibinfo {author} {\bibfnamefont {C.}~\bibnamefont {Gogolin}}, \bibinfo {author} {\bibfnamefont {J.}~\bibnamefont {Anders}},\ and\ \bibinfo {author} {\bibfnamefont {G.}~\bibnamefont {Adesso}},\ }\href {https://doi.org/https://link.springer.com/book/10.1007/978-3-319-99046-0} {\emph {\bibinfo {title} {Thermodynamics in the Quantum Regime: Fundamental Aspects and New Directions}}},\ Fundamental Theories of Physics\ (\bibinfo  {publisher} {Springer International Publishing},\ \bibinfo {year} {2019})\BibitemShut {NoStop}%
	\bibitem [{\citenamefont {Anshu}\ and\ \citenamefont {Arunachalam}(2024)}]{Anshu2024}%
	\BibitemOpen
	\bibfield  {author} {\bibinfo {author} {\bibfnamefont {A.}~\bibnamefont {Anshu}}\ and\ \bibinfo {author} {\bibfnamefont {S.}~\bibnamefont {Arunachalam}},\ }\bibfield  {title} {\bibinfo {title} {A survey on the complexity of learning quantum states},\ }\href {https://doi.org/10.1038/s42254-023-00662-4} {\bibfield  {journal} {\bibinfo  {journal} {Nat. Rev. Phys.}\ }\textbf {\bibinfo {volume} {6}},\ \bibinfo {pages} {59} (\bibinfo {year} {2024})}\BibitemShut {NoStop}%
	\bibitem [{\citenamefont {Keyl}\ and\ \citenamefont {Werner}(2001)}]{keyl_2001}%
	\BibitemOpen
	\bibfield  {author} {\bibinfo {author} {\bibfnamefont {M.}~\bibnamefont {Keyl}}\ and\ \bibinfo {author} {\bibfnamefont {R.~F.}\ \bibnamefont {Werner}},\ }\bibfield  {title} {\bibinfo {title} {Estimating the spectrum of a density operator},\ }\href {https://doi.org/10.1103/PhysRevA.64.052311} {\bibfield  {journal} {\bibinfo  {journal} {Phys. Rev. A}\ }\textbf {\bibinfo {volume} {64}},\ \bibinfo {pages} {052311} (\bibinfo {year} {2001})}\BibitemShut {NoStop}%
	\bibitem [{\citenamefont {Haah}\ \emph {et~al.}(2017)\citenamefont {Haah}, \citenamefont {Harrow}, \citenamefont {Ji}, \citenamefont {Wu},\ and\ \citenamefont {Yu}}]{harrow_17}%
	\BibitemOpen
	\bibfield  {author} {\bibinfo {author} {\bibfnamefont {J.}~\bibnamefont {Haah}}, \bibinfo {author} {\bibfnamefont {A.~W.}\ \bibnamefont {Harrow}}, \bibinfo {author} {\bibfnamefont {Z.}~\bibnamefont {Ji}}, \bibinfo {author} {\bibfnamefont {X.}~\bibnamefont {Wu}},\ and\ \bibinfo {author} {\bibfnamefont {N.}~\bibnamefont {Yu}},\ }\bibfield  {title} {\bibinfo {title} {Sample-optimal tomography of quantum states},\ }\href {https://doi.org/10.1109/TIT.2017.2719044} {\bibfield  {journal} {\bibinfo  {journal} {IEEE Trans. Inf. Theory}\ }\textbf {\bibinfo {volume} {63}},\ \bibinfo {pages} {5628} (\bibinfo {year} {2017})}\BibitemShut {NoStop}%
	\bibitem [{\citenamefont {O'Donnell}\ and\ \citenamefont {Wright}(2016)}]{odonnell_tomography}%
	\BibitemOpen
	\bibfield  {author} {\bibinfo {author} {\bibfnamefont {R.}~\bibnamefont {O'Donnell}}\ and\ \bibinfo {author} {\bibfnamefont {J.}~\bibnamefont {Wright}},\ }\bibfield  {title} {\bibinfo {title} {Efficient quantum tomography},\ }in\ \href {https://doi.org/10.1145/2897518.2897544} {\emph {\bibinfo {booktitle} {Proceedings of the Forty-Eighth Annual ACM Symposium on Theory of Computing}}},\ \bibinfo {series and number} {STOC '16}\ (\bibinfo  {publisher} {Association for Computing Machinery},\ \bibinfo {address} {New York, NY, USA},\ \bibinfo {year} {2016})\ p.\ \bibinfo {pages} {899–912}\BibitemShut {NoStop}%
	\bibitem [{\citenamefont {Massar}\ and\ \citenamefont {Popescu}(1995)}]{popescu_massar_95}%
	\BibitemOpen
	\bibfield  {author} {\bibinfo {author} {\bibfnamefont {S.}~\bibnamefont {Massar}}\ and\ \bibinfo {author} {\bibfnamefont {S.}~\bibnamefont {Popescu}},\ }\bibfield  {title} {\bibinfo {title} {Optimal extraction of information from finite quantum ensembles},\ }\href {https://doi.org/10.1103/PhysRevLett.74.1259} {\bibfield  {journal} {\bibinfo  {journal} {Phys. Rev. Lett.}\ }\textbf {\bibinfo {volume} {74}},\ \bibinfo {pages} {1259} (\bibinfo {year} {1995})}\BibitemShut {NoStop}%
	\bibitem [{\citenamefont {D'Ariano}\ \emph {et~al.}(2003)\citenamefont {D'Ariano}, \citenamefont {Paris},\ and\ \citenamefont {Sacchi}}]{d2003quantum}%
	\BibitemOpen
	\bibfield  {author} {\bibinfo {author} {\bibfnamefont {G.~M.}\ \bibnamefont {D'Ariano}}, \bibinfo {author} {\bibfnamefont {M.~G.}\ \bibnamefont {Paris}},\ and\ \bibinfo {author} {\bibfnamefont {M.~F.}\ \bibnamefont {Sacchi}},\ }\bibfield  {title} {\bibinfo {title} {Quantum tomography},\ }\href {https://doi.org/https://shop.elsevier.com/books/advances-in-imaging-and-electron-physics/hawkes/978-0-12-014770-0} {\bibfield  {journal} {\bibinfo  {journal} {Advances in imaging and electron physics}\ }\textbf {\bibinfo {volume} {128}},\ \bibinfo {pages} {S1076} (\bibinfo {year} {2003})}\BibitemShut {NoStop}%
	\bibitem [{\citenamefont {Watrous}(2018)}]{watrous2018theory}%
	\BibitemOpen
	\bibfield  {author} {\bibinfo {author} {\bibfnamefont {J.}~\bibnamefont {Watrous}},\ }\href {https://doi.org/https://doi.org/10.1017/9781316848142} {\emph {\bibinfo {title} {The theory of quantum information}}}\ (\bibinfo  {publisher} {Cambridge university press},\ \bibinfo {year} {2018})\BibitemShut {NoStop}%
	\bibitem [{\citenamefont {Peres}(1995)}]{peres1995quantum}%
	\BibitemOpen
	\bibfield  {author} {\bibinfo {author} {\bibfnamefont {A.}~\bibnamefont {Peres}},\ }\href {https://books.google.pl/books?id=rMGqMyFBcL8C} {\emph {\bibinfo {title} {Quantum Theory: Concepts and Methods}}}\ (\bibinfo  {publisher} {Springer},\ \bibinfo {year} {1995})\BibitemShut {NoStop}%
	\bibitem [{\citenamefont {Scarani}\ \emph {et~al.}(2005)\citenamefont {Scarani}, \citenamefont {Iblisdir}, \citenamefont {Gisin},\ and\ \citenamefont {Ac\'{\i}n}}]{scarani2005quantum}%
	\BibitemOpen
	\bibfield  {author} {\bibinfo {author} {\bibfnamefont {V.}~\bibnamefont {Scarani}}, \bibinfo {author} {\bibfnamefont {S.}~\bibnamefont {Iblisdir}}, \bibinfo {author} {\bibfnamefont {N.}~\bibnamefont {Gisin}},\ and\ \bibinfo {author} {\bibfnamefont {A.}~\bibnamefont {Ac\'{\i}n}},\ }\bibfield  {title} {\bibinfo {title} {Quantum cloning},\ }\href {https://doi.org/10.1103/RevModPhys.77.1225} {\bibfield  {journal} {\bibinfo  {journal} {Rev. Mod. Phys.}\ }\textbf {\bibinfo {volume} {77}},\ \bibinfo {pages} {1225} (\bibinfo {year} {2005})}\BibitemShut {NoStop}%
	\bibitem [{\citenamefont {Jaynes}(1957)}]{jaynes1957informationI}%
	\BibitemOpen
	\bibfield  {author} {\bibinfo {author} {\bibfnamefont {E.~T.}\ \bibnamefont {Jaynes}},\ }\bibfield  {title} {\bibinfo {title} {Information theory and statistical mechanics},\ }\href {https://doi.org/https://doi.org/10.1103/PhysRev.106.620} {\bibfield  {journal} {\bibinfo  {journal} {Phys. Rev}\ }\textbf {\bibinfo {volume} {106}},\ \bibinfo {pages} {620} (\bibinfo {year} {1957})}\BibitemShut {NoStop}%
	\bibitem [{\citenamefont {Gisin}\ and\ \citenamefont {Massar}(1997)}]{gisin_massar_97}%
	\BibitemOpen
	\bibfield  {author} {\bibinfo {author} {\bibfnamefont {N.}~\bibnamefont {Gisin}}\ and\ \bibinfo {author} {\bibfnamefont {S.}~\bibnamefont {Massar}},\ }\bibfield  {title} {\bibinfo {title} {Optimal quantum cloning machines},\ }\href {https://doi.org/10.1103/PhysRevLett.79.2153} {\bibfield  {journal} {\bibinfo  {journal} {Phys. Rev. Lett.}\ }\textbf {\bibinfo {volume} {79}},\ \bibinfo {pages} {2153} (\bibinfo {year} {1997})}\BibitemShut {NoStop}%
	\bibitem [{\citenamefont {Werner}(1998)}]{werner_98}%
	\BibitemOpen
	\bibfield  {author} {\bibinfo {author} {\bibfnamefont {R.~F.}\ \bibnamefont {Werner}},\ }\bibfield  {title} {\bibinfo {title} {Optimal cloning of pure states},\ }\href {https://doi.org/10.1103/PhysRevA.58.1827} {\bibfield  {journal} {\bibinfo  {journal} {Phys. Rev. A}\ }\textbf {\bibinfo {volume} {58}},\ \bibinfo {pages} {1827} (\bibinfo {year} {1998})}\BibitemShut {NoStop}%
	\bibitem [{\citenamefont {Bruss}\ \emph {et~al.}(1998)\citenamefont {Bruss}, \citenamefont {Ekert},\ and\ \citenamefont {Macchiavello}}]{bruss_ekert_98}%
	\BibitemOpen
	\bibfield  {author} {\bibinfo {author} {\bibfnamefont {D.}~\bibnamefont {Bruss}}, \bibinfo {author} {\bibfnamefont {A.}~\bibnamefont {Ekert}},\ and\ \bibinfo {author} {\bibfnamefont {C.}~\bibnamefont {Macchiavello}},\ }\bibfield  {title} {\bibinfo {title} {Optimal universal quantum cloning and state estimation},\ }\href {https://doi.org/10.1103/PhysRevLett.81.2598} {\bibfield  {journal} {\bibinfo  {journal} {Phys. Rev. Lett.}\ }\textbf {\bibinfo {volume} {81}},\ \bibinfo {pages} {2598} (\bibinfo {year} {1998})}\BibitemShut {NoStop}%
	\bibitem [{\citenamefont {Bruß}\ and\ \citenamefont {Macchiavello}(1999)}]{Bruss_1999}%
	\BibitemOpen
	\bibfield  {author} {\bibinfo {author} {\bibfnamefont {D.}~\bibnamefont {Bruß}}\ and\ \bibinfo {author} {\bibfnamefont {C.}~\bibnamefont {Macchiavello}},\ }\bibfield  {title} {\bibinfo {title} {Optimal state estimation for d-dimensional quantum systems},\ }\href {https://doi.org/10.1016/s0375-9601(99)00099-7} {\bibfield  {journal} {\bibinfo  {journal} {Phys. Lett. A}\ }\textbf {\bibinfo {volume} {253}},\ \bibinfo {pages} {249–251} (\bibinfo {year} {1999})}\BibitemShut {NoStop}%
	\bibitem [{\citenamefont {Bae}\ and\ \citenamefont {Ac\'{\i}n}(2006)}]{acin_06}%
	\BibitemOpen
	\bibfield  {author} {\bibinfo {author} {\bibfnamefont {J.}~\bibnamefont {Bae}}\ and\ \bibinfo {author} {\bibfnamefont {A.}~\bibnamefont {Ac\'{\i}n}},\ }\bibfield  {title} {\bibinfo {title} {Asymptotic quantum cloning is state estimation},\ }\href {https://doi.org/10.1103/PhysRevLett.97.030402} {\bibfield  {journal} {\bibinfo  {journal} {Phys. Rev. Lett.}\ }\textbf {\bibinfo {volume} {97}},\ \bibinfo {pages} {030402} (\bibinfo {year} {2006})}\BibitemShut {NoStop}%
	\bibitem [{\citenamefont {Marshall}\ \emph {et~al.}(1979)\citenamefont {Marshall}, \citenamefont {Olkin},\ and\ \citenamefont {Arnold}}]{marshall1979inequalities}%
	\BibitemOpen
	\bibfield  {author} {\bibinfo {author} {\bibfnamefont {A.~W.}\ \bibnamefont {Marshall}}, \bibinfo {author} {\bibfnamefont {I.}~\bibnamefont {Olkin}},\ and\ \bibinfo {author} {\bibfnamefont {B.~C.}\ \bibnamefont {Arnold}},\ }\href {https://link.springer.com/book/10.1007/978-0-387-68276-1} {\emph {\bibinfo {title} {Inequalities: theory of majorization and its applications}}},\ Vol.\ \bibinfo {volume} {143}\ (\bibinfo  {publisher} {Springer},\ \bibinfo {year} {1979})\BibitemShut {NoStop}%
	\bibitem [{\citenamefont {Szilard}(1929)}]{szilard1929entropieverminderung}%
	\BibitemOpen
	\bibfield  {author} {\bibinfo {author} {\bibfnamefont {L.}~\bibnamefont {Szilard}},\ }\bibfield  {title} {\bibinfo {title} {{\"U}ber die entropieverminderung in einem thermodynamischen system bei eingriffen intelligenter wesen},\ }\href {https://doi.org/10.1007/BF01341281} {\bibfield  {journal} {\bibinfo  {journal} {Z. Phys.}\ }\textbf {\bibinfo {volume} {53}},\ \bibinfo {pages} {840} (\bibinfo {year} {1929})}\BibitemShut {NoStop}%
	\bibitem [{\citenamefont {Dunlop}\ \emph {et~al.}(2025)\citenamefont {Dunlop}, \citenamefont {Cerisola}, \citenamefont {Monsel}, \citenamefont {Sevitz}, \citenamefont {Tabanera-Bravo}, \citenamefont {Dexter}, \citenamefont {Fedele}, \citenamefont {Ares},\ and\ \citenamefont {Anders}}]{cerisola_25}%
	\BibitemOpen
	\bibfield  {author} {\bibinfo {author} {\bibfnamefont {J.}~\bibnamefont {Dunlop}}, \bibinfo {author} {\bibfnamefont {F.}~\bibnamefont {Cerisola}}, \bibinfo {author} {\bibfnamefont {J.}~\bibnamefont {Monsel}}, \bibinfo {author} {\bibfnamefont {S.}~\bibnamefont {Sevitz}}, \bibinfo {author} {\bibfnamefont {J.}~\bibnamefont {Tabanera-Bravo}}, \bibinfo {author} {\bibfnamefont {J.}~\bibnamefont {Dexter}}, \bibinfo {author} {\bibfnamefont {F.}~\bibnamefont {Fedele}}, \bibinfo {author} {\bibfnamefont {N.}~\bibnamefont {Ares}},\ and\ \bibinfo {author} {\bibfnamefont {J.}~\bibnamefont {Anders}},\ }\bibfield  {title} {\bibinfo {title} {Extra cost of erasure due to quantum lifetime broadening},\ }\href {https://doi.org/10.1103/pc2t-ybtz} {\bibfield  {journal} {\bibinfo  {journal} {Phys. Rev. A}\ }\textbf {\bibinfo {volume} {112}},\ \bibinfo {pages} {L010601} (\bibinfo {year} {2025})}\BibitemShut {NoStop}%
	\bibitem [{\citenamefont {Xuereb}(2024)}]{code}%
	\BibitemOpen
	\bibfield  {author} {\bibinfo {author} {\bibfnamefont {J.}~\bibnamefont {Xuereb}},\ }\href@noop {} {\bibinfo {title} {Mathematica code for knowledge concentration example}},\ \bibinfo {howpublished} {\url{https://github.com/jqed-xuereb/When-is-Knowledge-Free-in-QThermo}} (\bibinfo {year} {2024})\BibitemShut {NoStop}%
	\bibitem [{\citenamefont {Brillouin}(1953)}]{brillouin1953negentropy}%
	\BibitemOpen
	\bibfield  {author} {\bibinfo {author} {\bibfnamefont {L.}~\bibnamefont {Brillouin}},\ }\bibfield  {title} {\bibinfo {title} {The negentropy principle of information},\ }\href {https://pubs.aip.org/aip/jap/article-abstract/24/9/1152/160687/The-Negentropy-Principle-of-Information?redirectedFrom=fulltext} {\bibfield  {journal} {\bibinfo  {journal} {J. Appl. Phys.}\ }\textbf {\bibinfo {volume} {24}},\ \bibinfo {pages} {1152} (\bibinfo {year} {1953})}\BibitemShut {NoStop}%
	\bibitem [{\citenamefont {Sagawa}\ and\ \citenamefont {Ueda}(2009)}]{PhysRevLett.102.250602}%
	\BibitemOpen
	\bibfield  {author} {\bibinfo {author} {\bibfnamefont {T.}~\bibnamefont {Sagawa}}\ and\ \bibinfo {author} {\bibfnamefont {M.}~\bibnamefont {Ueda}},\ }\bibfield  {title} {\bibinfo {title} {Minimal energy cost for thermodynamic information processing: Measurement and information erasure},\ }\href {https://doi.org/10.1103/PhysRevLett.102.250602} {\bibfield  {journal} {\bibinfo  {journal} {Phys. Rev. Lett.}\ }\textbf {\bibinfo {volume} {102}},\ \bibinfo {pages} {250602} (\bibinfo {year} {2009})}\BibitemShut {NoStop}%
	\bibitem [{\citenamefont {Guryanova}\ \emph {et~al.}(2020)\citenamefont {Guryanova}, \citenamefont {Friis},\ and\ \citenamefont {Huber}}]{Guryanova_2020}%
	\BibitemOpen
	\bibfield  {author} {\bibinfo {author} {\bibfnamefont {Y.}~\bibnamefont {Guryanova}}, \bibinfo {author} {\bibfnamefont {N.}~\bibnamefont {Friis}},\ and\ \bibinfo {author} {\bibfnamefont {M.}~\bibnamefont {Huber}},\ }\bibfield  {title} {\bibinfo {title} {Ideal projective measurements have infinite resource costs},\ }\href {https://doi.org/10.22331/q-2020-01-13-222} {\bibfield  {journal} {\bibinfo  {journal} {Quantum}\ }\textbf {\bibinfo {volume} {4}},\ \bibinfo {pages} {222} (\bibinfo {year} {2020})}\BibitemShut {NoStop}%
	\bibitem [{\citenamefont {Minagawa}\ \emph {et~al.}(2023)\citenamefont {Minagawa}, \citenamefont {Mohammady}, \citenamefont {Sakai}, \citenamefont {Kato},\ and\ \citenamefont {Buscemi}}]{minagawa2023universalvaliditysecondlaw}%
	\BibitemOpen
	\bibfield  {author} {\bibinfo {author} {\bibfnamefont {S.}~\bibnamefont {Minagawa}}, \bibinfo {author} {\bibfnamefont {M.~H.}\ \bibnamefont {Mohammady}}, \bibinfo {author} {\bibfnamefont {K.}~\bibnamefont {Sakai}}, \bibinfo {author} {\bibfnamefont {K.}~\bibnamefont {Kato}},\ and\ \bibinfo {author} {\bibfnamefont {F.}~\bibnamefont {Buscemi}},\ }\bibfield  {title} {\bibinfo {title} {Universal validity of the second law of information thermodynamics},\ }\href {https://arxiv.org/abs/2308.15558} {\bibfield  {journal} {\bibinfo  {journal} {arXiv:2308.15558}\ } (\bibinfo {year} {2023})}\BibitemShut {NoStop}%
	\bibitem [{\citenamefont {Debarba}\ \emph {et~al.}(2024)\citenamefont {Debarba}, \citenamefont {Huber},\ and\ \citenamefont {Friis}}]{Debarba2024}%
	\BibitemOpen
	\bibfield  {author} {\bibinfo {author} {\bibfnamefont {T.}~\bibnamefont {Debarba}}, \bibinfo {author} {\bibfnamefont {M.}~\bibnamefont {Huber}},\ and\ \bibinfo {author} {\bibfnamefont {N.}~\bibnamefont {Friis}},\ }\href@noop {} {\bibinfo {title} {Broadcasting quantum information using finite resources}} (\bibinfo {year} {2024}),\ \Eprint {https://arxiv.org/abs/2403.07660v1} {arXiv:2403.07660v1 [quant-ph]} \BibitemShut {NoStop}%
	\bibitem [{\citenamefont {Breuer}\ and\ \citenamefont {Petruccione}(2007)}]{petruccione}%
	\BibitemOpen
	\bibfield  {author} {\bibinfo {author} {\bibfnamefont {H.-P.}\ \bibnamefont {Breuer}}\ and\ \bibinfo {author} {\bibfnamefont {F.}~\bibnamefont {Petruccione}},\ }\href {https://doi.org/10.1093/acprof:oso/9780199213900.001.0001} {\emph {\bibinfo {title} {The Theory of Open Quantum Systems}}}\ (\bibinfo  {publisher} {Oxford University Press},\ \bibinfo {year} {2007})\BibitemShut {NoStop}%
	\bibitem [{\citenamefont {Marshall}\ \emph {et~al.}(2010)\citenamefont {Marshall}, \citenamefont {Olkin},\ and\ \citenamefont {Arnold}}]{marshall2010inequalities}%
	\BibitemOpen
	\bibfield  {author} {\bibinfo {author} {\bibfnamefont {A.}~\bibnamefont {Marshall}}, \bibinfo {author} {\bibfnamefont {I.}~\bibnamefont {Olkin}},\ and\ \bibinfo {author} {\bibfnamefont {B.}~\bibnamefont {Arnold}},\ }\href {https://books.google.at/books?id=I9wfajyOrooC} {\emph {\bibinfo {title} {Inequalities: Theory of Majorization and Its Applications}}},\ Springer Series in Statistics\ (\bibinfo  {publisher} {Springer New York},\ \bibinfo {year} {2010})\BibitemShut {NoStop}%
	\bibitem [{\citenamefont {Allahverdyan}\ \emph {et~al.}(2004)\citenamefont {Allahverdyan}, \citenamefont {Balian},\ and\ \citenamefont {Nieuwenhuizen}}]{allahverdyan2004maximal}%
	\BibitemOpen
	\bibfield  {author} {\bibinfo {author} {\bibfnamefont {A.~E.}\ \bibnamefont {Allahverdyan}}, \bibinfo {author} {\bibfnamefont {R.}~\bibnamefont {Balian}},\ and\ \bibinfo {author} {\bibfnamefont {T.~M.}\ \bibnamefont {Nieuwenhuizen}},\ }\bibfield  {title} {\bibinfo {title} {Maximal work extraction from finite quantum systems},\ }\href {https://doi.org/10.1209/epl/i2004-10101-2} {\bibfield  {journal} {\bibinfo  {journal} {EPL}\ }\textbf {\bibinfo {volume} {67}},\ \bibinfo {pages} {565} (\bibinfo {year} {2004})}\BibitemShut {NoStop}%
	\bibitem [{\citenamefont {\ifmmode~\check{S}\else \v{S}\fi{}afr\'anek}\ \emph {et~al.}(2023)\citenamefont {\ifmmode~\check{S}\else \v{S}\fi{}afr\'anek}, \citenamefont {Rosa},\ and\ \citenamefont {Binder}}]{safranek_2023}%
	\BibitemOpen
	\bibfield  {author} {\bibinfo {author} {\bibfnamefont {D.}~\bibnamefont {\ifmmode~\check{S}\else \v{S}\fi{}afr\'anek}}, \bibinfo {author} {\bibfnamefont {D.}~\bibnamefont {Rosa}},\ and\ \bibinfo {author} {\bibfnamefont {F.~C.}\ \bibnamefont {Binder}},\ }\bibfield  {title} {\bibinfo {title} {Work extraction from unknown quantum sources},\ }\href {https://doi.org/10.1103/PhysRevLett.130.210401} {\bibfield  {journal} {\bibinfo  {journal} {Phys. Rev. Lett.}\ }\textbf {\bibinfo {volume} {130}},\ \bibinfo {pages} {210401} (\bibinfo {year} {2023})}\BibitemShut {NoStop}%
	\bibitem [{\citenamefont {Chitambar}\ and\ \citenamefont {Gour}(2019)}]{chitambarGour}%
	\BibitemOpen
	\bibfield  {author} {\bibinfo {author} {\bibfnamefont {E.}~\bibnamefont {Chitambar}}\ and\ \bibinfo {author} {\bibfnamefont {G.}~\bibnamefont {Gour}},\ }\bibfield  {title} {\bibinfo {title} {Quantum resource theories},\ }\href {https://doi.org/10.1103/RevModPhys.91.025001} {\bibfield  {journal} {\bibinfo  {journal} {Rev. Mod. Phys.}\ }\textbf {\bibinfo {volume} {91}},\ \bibinfo {pages} {025001} (\bibinfo {year} {2019})}\BibitemShut {NoStop}%
	\bibitem [{\citenamefont {Janzing}\ \emph {et~al.}(2000)\citenamefont {Janzing}, \citenamefont {Wocjan}, \citenamefont {Zeier}, \citenamefont {Geiss},\ and\ \citenamefont {Beth}}]{Janzing2000}%
	\BibitemOpen
	\bibfield  {author} {\bibinfo {author} {\bibfnamefont {D.}~\bibnamefont {Janzing}}, \bibinfo {author} {\bibfnamefont {P.}~\bibnamefont {Wocjan}}, \bibinfo {author} {\bibfnamefont {R.}~\bibnamefont {Zeier}}, \bibinfo {author} {\bibfnamefont {R.}~\bibnamefont {Geiss}},\ and\ \bibinfo {author} {\bibfnamefont {T.}~\bibnamefont {Beth}},\ }\bibfield  {title} {\bibinfo {title} {Thermodynamic cost of reliability and low temperatures: Tightening landauer's principle and the second law},\ }\href {https://doi.org/10.1023/A:1026422630734} {\bibfield  {journal} {\bibinfo  {journal} {Int. J. Theor. Phys.}\ }\textbf {\bibinfo {volume} {39}},\ \bibinfo {pages} {2717} (\bibinfo {year} {2000})}\BibitemShut {NoStop}%
	\bibitem [{\citenamefont {{Horodecki}}\ and\ \citenamefont {{Oppenheim}}(2013)}]{horodecki2013fundamental}%
	\BibitemOpen
	\bibfield  {author} {\bibinfo {author} {\bibfnamefont {M.}~\bibnamefont {{Horodecki}}}\ and\ \bibinfo {author} {\bibfnamefont {J.}~\bibnamefont {{Oppenheim}}},\ }\bibfield  {title} {\bibinfo {title} {{Fundamental limitations for quantum and nanoscale thermodynamics}},\ }\href {https://www.nature.com/articles/ncomms3059} {\bibfield  {journal} {\bibinfo  {journal} {Nat. Commun.}\ }\textbf {\bibinfo {volume} {4}},\ \bibinfo {eid} {2059} (\bibinfo {year} {2013})}\BibitemShut {NoStop}%
	\bibitem [{\citenamefont {{Brand\~ao}}\ \emph {et~al.}(2015)\citenamefont {{Brand\~ao}}, \citenamefont {{Horodecki}}, \citenamefont {{Ng}}, \citenamefont {{Oppenheim}},\ and\ \citenamefont {{Wehner}}}]{brandao2015second}%
	\BibitemOpen
	\bibfield  {author} {\bibinfo {author} {\bibfnamefont {F.~G.~S.~L.}\ \bibnamefont {{Brand\~ao}}}, \bibinfo {author} {\bibfnamefont {M.}~\bibnamefont {{Horodecki}}}, \bibinfo {author} {\bibfnamefont {N.~H.~Y.}\ \bibnamefont {{Ng}}}, \bibinfo {author} {\bibfnamefont {J.}~\bibnamefont {{Oppenheim}}},\ and\ \bibinfo {author} {\bibfnamefont {S.}~\bibnamefont {{Wehner}}},\ }\bibfield  {title} {\bibinfo {title} {{The second laws of quantum thermodynamics}},\ }\href {https://doi.org/10.1073/pnas.1411728112} {\bibfield  {journal} {\bibinfo  {journal} {Proc. Natl. Acad. Sci. U.S.A.}\ }\textbf {\bibinfo {volume} {112}},\ \bibinfo {pages} {3275} (\bibinfo {year} {2015})}\BibitemShut {NoStop}%
	\bibitem [{\citenamefont {Lostaglio}(2019)}]{Lostaglio2019}%
	\BibitemOpen
	\bibfield  {author} {\bibinfo {author} {\bibfnamefont {M.}~\bibnamefont {Lostaglio}},\ }\bibfield  {title} {\bibinfo {title} {An introductory review of the resource theory approach to thermodynamics},\ }\href {https://doi.org/10.1088/1361-6633/ab46e5} {\bibfield  {journal} {\bibinfo  {journal} {Rep. Prog. Phys.}\ }\textbf {\bibinfo {volume} {82}},\ \bibinfo {pages} {114001} (\bibinfo {year} {2019})}\BibitemShut {NoStop}%
	\bibitem [{\citenamefont {de~Oliveira~Junior}\ \emph {et~al.}(2022)\citenamefont {de~Oliveira~Junior}, \citenamefont {Czartowski}, \citenamefont {{\.Z}yczkowski},\ and\ \citenamefont {Korzekwa}}]{deoliveirajunior2022}%
	\BibitemOpen
	\bibfield  {author} {\bibinfo {author} {\bibfnamefont {A.}~\bibnamefont {de~Oliveira~Junior}}, \bibinfo {author} {\bibfnamefont {J.}~\bibnamefont {Czartowski}}, \bibinfo {author} {\bibfnamefont {K.}~\bibnamefont {{\.Z}yczkowski}},\ and\ \bibinfo {author} {\bibfnamefont {K.}~\bibnamefont {Korzekwa}},\ }\bibfield  {title} {\bibinfo {title} {Geometric structure of thermal cones},\ }\href {https://doi.org/10.1103/PhysRevE.106.064109} {\bibfield  {journal} {\bibinfo  {journal} {Phys. Rev. E}\ }\textbf {\bibinfo {volume} {106}},\ \bibinfo {pages} {064109} (\bibinfo {year} {2022})}\BibitemShut {NoStop}%
	\bibitem [{\citenamefont {Landi}\ and\ \citenamefont {Paternostro}(2021)}]{landi_entropy}%
	\BibitemOpen
	\bibfield  {author} {\bibinfo {author} {\bibfnamefont {G.~T.}\ \bibnamefont {Landi}}\ and\ \bibinfo {author} {\bibfnamefont {M.}~\bibnamefont {Paternostro}},\ }\bibfield  {title} {\bibinfo {title} {Irreversible entropy production: From classical to quantum},\ }\href {https://doi.org/10.1103/RevModPhys.93.035008} {\bibfield  {journal} {\bibinfo  {journal} {Rev. Mod. Phys.}\ }\textbf {\bibinfo {volume} {93}},\ \bibinfo {pages} {035008} (\bibinfo {year} {2021})}\BibitemShut {NoStop}%
	\bibitem [{\citenamefont {Esposito}\ \emph {et~al.}(2010)\citenamefont {Esposito}, \citenamefont {Lindenberg},\ and\ \citenamefont {den Broeck}}]{Esposito_2010}%
	\BibitemOpen
	\bibfield  {author} {\bibinfo {author} {\bibfnamefont {M.}~\bibnamefont {Esposito}}, \bibinfo {author} {\bibfnamefont {K.}~\bibnamefont {Lindenberg}},\ and\ \bibinfo {author} {\bibfnamefont {C.~V.}\ \bibnamefont {den Broeck}},\ }\bibfield  {title} {\bibinfo {title} {Entropy production as correlation between system and reservoir},\ }\href {https://doi.org/10.1088/1367-2630/12/1/013013} {\bibfield  {journal} {\bibinfo  {journal} {New J. Phys.}\ }\textbf {\bibinfo {volume} {12}},\ \bibinfo {pages} {013013} (\bibinfo {year} {2010})}\BibitemShut {NoStop}%
	\bibitem [{\citenamefont {Reeb}\ and\ \citenamefont {Wolf}(2014)}]{Reeb_2014}%
	\BibitemOpen
	\bibfield  {author} {\bibinfo {author} {\bibfnamefont {D.}~\bibnamefont {Reeb}}\ and\ \bibinfo {author} {\bibfnamefont {M.~M.}\ \bibnamefont {Wolf}},\ }\bibfield  {title} {\bibinfo {title} {An improved landauer principle with finite-size corrections},\ }\href {https://doi.org/10.1088/1367-2630/16/10/103011} {\bibfield  {journal} {\bibinfo  {journal} {New J. Phys.}\ }\textbf {\bibinfo {volume} {16}},\ \bibinfo {pages} {103011} (\bibinfo {year} {2014})}\BibitemShut {NoStop}%
	\bibitem [{\citenamefont {Taranto}\ \emph {et~al.}(2020)\citenamefont {Taranto}, \citenamefont {Bakhshinezhad}, \citenamefont {Sch\"uttelkopf}, \citenamefont {Clivaz},\ and\ \citenamefont {Huber}}]{taranto_2020}%
	\BibitemOpen
	\bibfield  {author} {\bibinfo {author} {\bibfnamefont {P.}~\bibnamefont {Taranto}}, \bibinfo {author} {\bibfnamefont {F.}~\bibnamefont {Bakhshinezhad}}, \bibinfo {author} {\bibfnamefont {P.}~\bibnamefont {Sch\"uttelkopf}}, \bibinfo {author} {\bibfnamefont {F.}~\bibnamefont {Clivaz}},\ and\ \bibinfo {author} {\bibfnamefont {M.}~\bibnamefont {Huber}},\ }\bibfield  {title} {\bibinfo {title} {Exponential improvement for quantum cooling through finite-memory effects},\ }\href {https://doi.org/10.1103/PhysRevApplied.14.054005} {\bibfield  {journal} {\bibinfo  {journal} {Phys. Rev. Appl.}\ }\textbf {\bibinfo {volume} {14}},\ \bibinfo {pages} {054005} (\bibinfo {year} {2020})}\BibitemShut {NoStop}%
	\bibitem [{\citenamefont {Buffoni}\ and\ \citenamefont {Campisi}(2023)}]{Buffoni2023cooperativequantum}%
	\BibitemOpen
	\bibfield  {author} {\bibinfo {author} {\bibfnamefont {L.}~\bibnamefont {Buffoni}}\ and\ \bibinfo {author} {\bibfnamefont {M.}~\bibnamefont {Campisi}},\ }\bibfield  {title} {\bibinfo {title} {Cooperative quantum information erasure},\ }\href {https://doi.org/10.22331/q-2023-03-23-961} {\bibfield  {journal} {\bibinfo  {journal} {{Quantum}}\ }\textbf {\bibinfo {volume} {7}},\ \bibinfo {pages} {961} (\bibinfo {year} {2023})}\BibitemShut {NoStop}%
	\bibitem [{\citenamefont {Danageozian}\ \emph {et~al.}(2022)\citenamefont {Danageozian}, \citenamefont {Wilde},\ and\ \citenamefont {Buscemi}}]{wilde_triple_trade_off}%
	\BibitemOpen
	\bibfield  {author} {\bibinfo {author} {\bibfnamefont {A.}~\bibnamefont {Danageozian}}, \bibinfo {author} {\bibfnamefont {M.~M.}\ \bibnamefont {Wilde}},\ and\ \bibinfo {author} {\bibfnamefont {F.}~\bibnamefont {Buscemi}},\ }\bibfield  {title} {\bibinfo {title} {Thermodynamic constraints on quantum information gain and error correction: A triple trade-off},\ }\href {https://doi.org/10.1103/PRXQuantum.3.020318} {\bibfield  {journal} {\bibinfo  {journal} {PRX Quantum}\ }\textbf {\bibinfo {volume} {3}},\ \bibinfo {pages} {020318} (\bibinfo {year} {2022})}\BibitemShut {NoStop}%
\end{thebibliography}
\end{document}